\documentclass[a4paper,11pt]{article}

\usepackage{latexsym}
\usepackage{mathrsfs}
\usepackage[english]{babel}
\usepackage{amsmath}
\usepackage{amssymb}
\usepackage{amsfonts}

\usepackage{color}
\usepackage{graphicx}
\usepackage{enumerate}
\usepackage{bm}
\usepackage{cite}
\usepackage{url}
\usepackage{algorithm} 
\usepackage{algorithmic} 
\usepackage{subfigure}
\usepackage{mathtools}
\usepackage[mathcal]{euscript}

\usepackage{breqn}  

\newcommand{\R}{\mathbb{R}}

\newtheorem{theorem}{Theorem}
\newtheorem{definition}{Definition}
\newtheorem{corollary}{Corollary}
\newtheorem{lemma}{Lemma}
\newtheorem{remark}{Remark}
\newtheorem{proposition}{Proposition}

\newenvironment{proof}{\medskip\noindent{\it Proof. }}{ \medskip}

\DeclareFontFamily{OT1}{pzc}{}
\DeclareFontShape{OT1}{pzc}{m}{it}{<-> s * [1.200] pzcmi7t}{}
\DeclareMathAlphabet{\mathpzc}{OT1}{pzc}{m}{it}

\newcommand*\mcapinn[2]{\vcenter{\hbox{$\mathsurround=0pt
  \ifx\displaystyle#1\textstyle\else#1\fi\bigcap$}}}

\newcommand*\mcupinn[2]{\vcenter{\hbox{$\mathsurround=0pt
  \ifx\displaystyle#1\textstyle\else#1\fi\bigcup$}}}

\def\begequarr{\begin{eqnarray}}
\def\endequarr{\end{eqnarray}}
\def\begequarrs{\begin{eqnarray*}}
\def\endequarrs{\end{eqnarray*}}
\def\begequ{\begin{equation}}
\def\endequ{\end{equation}}
\def\begequs{\begin{equation*}}
\def\endequs{\end{equation*}}
\def\begite{\begin{itemize}}
\def\endite{\end{itemize}}

\def\begcen{\begin{center}}
\def\endcen{\end{center}}
\def\begrem{\begin{remark}\rm}
\def\endrem{\end{remark}}
\def\ba{\begin{array}}
\def\ea{\end{array}}

\def\rank{\textnormal{rank}\;}

\def\col{\textnormal{col}\; }

\def\det{\textnormal{det}}
\def\d{\textnormal{d}}

\def\Real{\mathbb R}

\def\dst{\displaystyle}

\def\pdf{\textnormal{pdf}}
\def\range{\textnormal{range}}
\def\span{\textnormal{span}}
\def\col{\textnormal{col}}

\def\Pbb{\mathbb{P}}

\newcommand{\V}{\mathrm{V}}
\newcommand{\G}{\mathrm{G}}
\newcommand{\E}{\mathrm{E}}

\newcommand{\ub}{\mathbf{u}}

\newcommand{\xb}{\mathbf{x}}
\newcommand{\yb}{\mathbf{y}}
\newcommand{\zb}{\mathbf{z}}

\newcommand{\etab}{\bm{\eta}}

\newcommand{\Hb}{\mathbf{H}}
\newcommand{\Lb}{\mathbf{L}}
\newcommand{\Ib}{\mathbf{I}}

\newcommand{\Fb}{\mathbf{F}}
\newcommand{\Eb}{\mathbf{E}}

\newcommand{\Kb}{\mathbf{K}}
\newcommand{\Ab}{\mathbf{A}}
\newcommand{\Bb}{\mathbf{B}}
\newcommand{\Cb}{\mathbf{C}}
\newcommand{\Db}{\mathbf{D}}

\newcommand{\eb}{\mathbf{e}}

\newcommand{\Xb}{\mathbf{X}}
\newcommand{\Yb}{\mathbf{Y}}
\newcommand{\Gb}{\mathbf{G}}
\newcommand{\Vb}{\mathbf{V}}
\newcommand{\Wb}{\mathbf{W}}
\newcommand{\Zb}{\mathbf{Z}}

\newcommand{\omegab}{\bm{\omega}}
\newcommand{\nub}{\bm{\nu}}

\def\beeq#1{\begin{equation}{#1}\end{equation}}

\def\baselinestretch{1.3}
\begin{document}

\title{\Huge Initial-Value Privacy  of  Linear  Dynamical Systems}
\author{Lei Wang, Ian R. Manchester, Jochen Trumpf and Guodong Shi
\thanks{*This research is supported by the Australian Research Council Discovery Project DP190103615. A preliminary work is scheduled to be presented at the 59th IEEE Conference on Decision and Control \cite{Wang-2020cdc}.}
\thanks{L. Wang, I. R. Manchester and G. Shi are with Australian Center for Field Robotics, The University of Sydney, Australia.
 (E-mail: lei.wang2; ian.manchester; guodong.shi@sydney.edu.au)}%
\thanks{ J. Trumpf is with College of Engineering and Computer Science, The Australian National University, Australia.
 (E-mail: Jochen.Trumpf@anu.edu.au)}
        }
\date{}
\maketitle

\begin{abstract}
This paper studies initial-value privacy problems of linear dynamical systems. We consider a standard linear time-invariant system with random process and measurement noises.  For such a system, eavesdroppers  having access to  system output trajectories may infer the system initial states, leading to initial-value privacy risks. When a finite number of  output trajectories are eavesdropped, we consider a requirement that any guess about the initial values  can be plausibly denied. When an infinite number of  output trajectories are  eavesdropped, we consider a requirement that the initial values should not be uniquely recoverable. In view of
these two privacy requirements, we define differential initial-value privacy and intrinsic initial-value privacy, respectively, for the system as metrics of privacy risks. First of all, we prove that the intrinsic initial-value privacy is equivalent to unobservability, while the  differential initial-value privacy can  be  achieved for a privacy budget  depending on an extended observability matrix of the system and the covariance of the noises.
Next, the inherent network nature of the considered linear system is explored, where each individual state corresponds to a node and the state and output matrices induce interaction and sensing graphs, leading to a network system. Under this network system perspective,  we allow the initial states at some nodes to be  public, and investigate the resulting intrinsic initial-value privacy of each individual node. We establish necessary and sufficient conditions for such individual node initial-value privacy, and also prove that the intrinsic initial-value privacy of individual nodes is generically determined by the network structure. These results may be  extended to linear systems with time-varying dynamics under the same analysis framework.

\end{abstract}


\section{Introduction}
The rapid developments  in networked control systems \cite{Hespanha2007}, internet of things \cite{Atzori2010,Gubbi2013}, smart grids \cite{actchallenge}, intelligent transportation \cite{Papadimitratos2009,Zhang2011} during the past decade shed lights on how future  infrastructures of our society can be made {\it smart} via interconnected sensing, dynamics, and control over cyber-physical systems. The operation of such systems inherently relies on users and subsystems sharing  signals such as  measurements, dynamical states,  control inputs in their local views, so that collective decisions become possible. The shared signals might directly contain  sensitive  information of a private nature, e.g.,   loads and currents in a  grid reflect directly activities in a residence or productions in a company \cite{Sandberg-CDC-2015}; or they might indirectly  encode physical parameters, user preferences,  economic  inclination, e.g., control inputs in economic model predictive control implicitly carry information about the system's economic objective as it is used  as the  objective function \cite{Ma-EMPC2012}.

Several privacy metrics have been developed to address privacy expectations of dynamical systems.
A notable metric is  differential privacy, which originated in computer science\cite{Dwork-1,Dwork-2,Roth-2010}. When a mechanism is applied  taking  the sensitive information as input   and producing  an output as the learning outcome, differential privacy guarantees plausible deniability of  any inference about the private information by eavesdroppers having access to the output. Differential privacy has become the canonical solutions for privacy risk characterization in  dataset processing, due to its quantitative nature and robustness to post-processing and side information \cite{Dwork-2,Roth-2010}. The differential privacy framework has also been generalized to  dynamical systems for problems ranging from  average consensus seeking  \cite{huang2012differentially} and  distributed optimization \cite{pappas2017,cortes2018tcns} to estimation and filtering  \cite{Ny-Pappas-TAC-2014,Sandberg-CDC-2015} and
feedback control \cite{Kawano-Cao-19,Hale-LQG}. Consistent with its root, differential privacy for a dynamical system  provides the system with the  ability to have plausible deniability facing eavesdroppers, e.g., recent surveys  in \cite{cortes2016cdc,HanPappas2018}.

Besides differential privacy, another related but different privacy risk lies in the possibility that an eavesdropper makes an accurate enough  estimation of the sensitive parameters or signals, perhaps from  a number of repeated  observations.   In \cite{mo2017privacy},  the variance matrix of  the maximum likelihood estimation was utilized  to measure how accurate the initial node states in a consensus network maybe estimated from the trajectories of one or more malicious nodes.  In \cite{Farokhi-2019}, a measure of privacy was developed using the inverse of the trace of the Fisher information matrix, which is a lower bound of the variance of estimation error of unbiased estimators.


In particular,  initial values of a dynamical system may carry sensitive private information, leading to privacy risks related to the initial values. For instance, when distributed load shedding in micro-grid systems is performed by employing an average consensus dynamics,   initial values represent  load demands of individual users \cite{load-shedding}.  In \cite{huang2012differentially,mo2017privacy}, the initial-value privacy of the average consensus algorithm over dynamical networks was studied, and injecting exponentially decaying noises was used as a privacy-protection approach. The privacy of the initial value for a dynamical system is also of significant theoretical interest as the   system trajectories or distributions of the system trajectories are fully parametrized by the initial value, in the presence of the plant knowledge.    In \cite{Farokhi-2019}, the initial-value privacy of a linear system was studied, and an optimal privacy-preserving policy was established for the probability density function of the additive noise such that the balance between the Fisher information-based privacy level and output performance is achieved.



In this paper, we study initial-value privacy problems of linear dynamical systems in the presence of random process and sensor noises.  For such a system, eavesdroppers  having access to  system output trajectories may infer the system initial states. When a finite number of output trajectories are eavesdropped, we consider a requirement that any guess about the initial values  can be plausibly denied. When an infinite number of  output trajectories are  eavesdropped, we consider a requirement that the initial values should not be uniquely recoverable. These requirements inspire us to define and investigate  two initial-value privacy metrics for the considered linear system: differential initial-value privacy on the plausible deniability, and intrinsic initial-value privacy on the fundamental  non-identifiability.
Next, we turn to the  inherent network nature of  linear systems, where each  dimension of the system state corresponds to a node, and the state  and output matrices induce interaction and sensing graphs. In the presence of malicious users or additional observations,   the initial states at a subset of the nodes may be known to the eavesdroppers as well.  With such a public disclosure set, the   intrinsic initial-value privacy of each individual node, and the structural privacy metric of the entire network, become  interesting and challenging questions.


 The main results  of this paper are summarized in the following:
\begin{itemize}
  \item For general linear systems, we prove that  intrinsic initial-value privacy is equivalent to unobservability; and that  differential initial-value privacy can   be achieved for a privacy budget  depending on an extended observability matrix of the system and the covariance of the noises.
  \item For networked linear systems, we   establish necessary and sufficient  conditions for intrinsic initial-value privacy of individual nodes, in the presence of a public disclosure set consisting of nodes with known initial states.  We also show  that  the network structure plays a generic role in determining the privacy of each node's initial value, and the maximally allowed number of arbitrary disclosed nodes under privacy guarantee as a network privacy index.
\end{itemize}
These results may be  extended to linear (network) systems with time-varying dynamics under the same analysis framework. The network privacy as proven to be a  generic structural property, is a generalization to  the classical structural observability results.

A preliminary version of the results is presented in \cite{Wang-2020cdc} where the technical  proofs, illustrative examples, and many  discussions  were not included.  The remainder of the paper is organized as follows. Section 2 formulates the problem of interest for linear dynamical systems. In Section 3, intrinsic initial-value privacy and differential privacy are explicitly defined and studied by regarding all initial values as a whole. Then regarding the system from the network system perspective, Section 4 analyzes the intrinsic initial-value privacy of individual nodes with a public disclosure set and studies a quantitative network privacy index, from exact and generic perspectives. Finally a brief conclusion is made in Section 5. All technical proofs are presented in the Appendix.

\medskip
\noindent
{\bf Notations}. We denote by $\R$ the real numbers and $\R^n$  the real space of $n$ dimension for any positive integer $n$. For a vector $x\in\R^n$, the norm $\|x\|=(x^\top x)^{\frac{1}{2}}$. For any $x_1,\ldots,x_m\in\R^n$, we denote $[x_1;\ldots;x_m]$ as a vector $\begin{bmatrix}
   x_1^\top &
   \ldots &
   x_m^\top
 \end{bmatrix}^\top\in\R^{mn}$, and $[x_1,\ldots,x_m]$ as a matrix of which the $i$-the column is $x_i$, $i=1,\ldots,m$. For any square matrix $A$, let $\sigma(A)$ denote the set of all eigenvalues of $A$, and $\sigma_M(A),\sigma_m(A)$ denote the maximum and minimum eigenvalues, respectively. For any matrix $A\in\R^{n\times m}$, the norm $\|A\|=\sigma_M(A^\top A)^{\frac{1}{2}}$. For any set $X\in\R^n$, we let $1_X(x)$ be a characteristic function, satisfying $1_X(x)=1$ for $x\in X$ and $1_X(x)=0$ for $x\notin X$. The range of a matrix or a function is denoted as $\range(\cdot)$, and the span of a matrix is denoted as $\span(\cdot)$. We denote $\pdf(\cdot)$ as the probability density function, and $\eb_{i}\in\R^n$ as a vector whose entries are all zero except the $i$-th being one.

\section{Problem Statement}

\subsection{Initial-Value Privacy for Linear Systems}
We consider the following linear time-invariant (LTI) system
\beeq{\label{eq:ini-sys}\ba{rcl}
\xb_{t+1} &=& \Ab\, \xb_t + \nub_t\,\\
\yb_{t} &=& \Cb\,\xb_t + \omegab_t
\ea}
for $t=0,1,\ldots$, where $\xb_t\in\mathbb{R}^n$ is state, $\yb_t\in\mathbb{R}^m$ is output,  $\nub_t\in\Real^n$ is process noise,  and  $\omegab_t\in\Real^m$ is measurement noise.
Throughout this paper, we assume that $\nub_t$ and $\omegab_t$ are random variables according to some zero-mean distributions, and $\rank(\Cb) > 0$.

In this paper, we suppose that initial values $\xb_0$ are privacy-sensitive information for the system. Eavesdroppers having access to the output trajectory $(\yb_t)_{t=0}^T$ with $T\geq n-1$ attempt to infer the private initial values.
To facilitate subsequent analysis, we  denote the measurement vector $\mathbf{Y}_t=[\yb_{T-t};\yb_{T-t+1};\ldots;\yb_T]$, the noise vectors $\mathbf{V}_t=[\nub_{T-t};\nub_{T-t+1};\ldots;\nub_{T-1}]$ and $\mathbf{W}_t=[\omegab_{T-t};\omegab_{T-t+1};\ldots;\omegab_T]$, and let
\[\ba{l}
\mathbf{O}_{\textnormal ob}= \begin{bmatrix}
                 \Cb \\
                 \Cb\Ab \\
                 \vdots \\
                 \Cb\Ab^{n-1}
               \end{bmatrix}\,,\quad
\mathbf{O}_t= \begin{bmatrix}
                 \Cb \\
                 \Cb\Ab \\
                 \vdots \\
                 \Cb\Ab^t
               \end{bmatrix}\,,\quad\\
\mathbf{H}_t = \begin{bmatrix}
          0 & 0 & \cdots & 0 & 0 \\
          \Cb & 0 & \ddots & 0 & 0 \\
          \Cb\Ab & \Cb & \ddots & 0 & 0 \\
          \vdots & \ddots & \ddots & \ddots & \vdots \\
          \Cb\Ab^{t-2} & \Cb\Ab^{t-3} & \ddots & \Cb & 0 \\
          \Cb\Ab^{t-1} & \Cb\Ab^{t-2} & \cdots & \Cb\Ab & \Cb
        \end{bmatrix}\,.
\ea\]
Here $\mathbf{O}_{\textnormal ob}$ is observability matrix, and $\mathbf{O}_t$ denotes   \emph{extended} observability matrix for $t\geq n$ and $\mathbf{H}_t$ is a lower block triangular Toeplitz matrix.
Thus, the mapping from initial state $\xb_0$ to the output trajectory $\mathbf{Y}_{T}$ as $\mathcal{M}:\mathbb{R}^n\rightarrow\mathbb{R}^{m(T+1)}$ can be  described by
\beeq{\label{eq:M}
\mathbf{Y}_T = \mathcal{M}(\xb_0) := \mathbf{O}_T\xb_0 + \mathbf{H}_T \mathbf{V}_T + \mathbf{W}_T\,.
}

The system (\ref{eq:ini-sys}) may be implemented or run independently for multiple times with the same initial state $\xb_0$. When all resulting output trajectories are eavesdropped, the eavesdropper may derive an estimate of $\xb_0$ by statistical inference methods such as maximum likelihood estimation (MLE). The resulting estimate accuracy may converge to zero as the number of eavesdropped output trajectories converges to infinity, leading to initial-value privacy risks.  In view of this, we consider a requirement that
\begin{itemize}
    \item[(R1)] the initial values should not be uniquely recoverable by an eavesdropper having an infinite number of  output trajectories.
\end{itemize}
To address the requirement (R1), we define intrinsic initial-value privacy as below.

\medskip
\begin{definition}\label{Definition:IP}
The system (\ref{eq:ini-sys}) preserves \emph{intrinsic initial-value privacy} if the initial state $\xb_0$ is statistically non-identifiable from observing $(\yb_t)_{t=0}^T$, i.e., for any $\xb_0\in\mathbb{R}^n$, there exists a $\xb^\prime_0\neq \xb_0\in\mathbb{R}^n$ such that
\begin{equation}\label{eq:IPiv}
  \pdf\left(\mathbf{Y}_T| \xb_0\right) = \pdf\left(\mathbf{Y}_T | \xb^\prime_0\right)\,.
\end{equation}
\end{definition}

In Definition \ref{Definition:IP}, equality (\ref{eq:IPiv}) indicates that there exist other values $\xb^\prime_0$, yielding the same output trajectory distribution as that of the initial value $\xb_0$. This in turn guarantees that the system preserving the intrinsic initial-value privacy satisfies the  requirement (R1).

\medskip
\begin{remark}
The intrinsic initial-value privacy guarantees the initial state $\xb_0$ indistinguishable from the $\xb^\prime_0$ satisfying (\ref{eq:IPiv}), which is related to the notion of undetectable attacks in the secure control literature, e.g. \cite{Pasqualetti-2015}, where the attacker tries to inject signals that are indistinguishable.
\end{remark}

\medskip
On the other hand, when a finite $N$ output trajectories are eavesdropped, the eavesdroppers may infer the initial values under which there is a probability of generating these output trajectories. In view of this, we consider a requirement that
\begin{itemize}
    \item[(R2)] any inference about the true initial value from the eavesdroppers can be denied by supplying any value within a range to the inference with a similar probability  of generating the eavesdropped $N$ output trajectories.
\end{itemize}
This property is referred to as plausible deniability in the literature \cite{Dwork-3}. With this in mind, we denote the  eavesdropped output trajectories as $\mathbf{Y}_T^1,\ldots,\mathbf{Y}_T^N$. The mapping from initial state $\xb_0$ to $\mathbf{Y}_T^1,\ldots,\mathbf{Y}_T^N$ is a concatenation of $N$ mappings $\mathcal{M}(\xb_0)$, i.e.,
\beeq{\label{eq:M^N}
\begin{bmatrix}
\mathbf{Y}_T^1 \cr \vdots\cr\mathbf{Y}_T^N
\end{bmatrix}
= \mathcal{M}^N(\xb_0) : =
\begin{bmatrix}
\mathcal{M}(\xb_0) \cr \vdots\cr \mathcal{M}(\xb_0)
\end{bmatrix}\,.
}
We then  define differential initial-value privacy as below \cite{Dwork-1,Dwork-2}.

\medskip
\begin{definition} \label{Definition:DP}
We define two initial values $\xb_0,\xb^\prime_0\in\mathbb{R}^n$ as $d$-adjacent if $\|\xb_0-\xb^\prime_0\|\leq d$.
The system (\ref{eq:ini-sys}) preserves \emph{$(\epsilon,\delta)$-differential privacy of initial values} for some privacy budgets $\epsilon>0, 0.5>\delta>0$ under $d$-adjacency if for all $R\subset\mbox{range}(\mathcal{M}^N)$,
\beeq{\label{eq:DPiv}
\Pbb(\mathcal{M}^N(\xb_0)\in R) \leq e^\epsilon\cdot \Pbb(\mathcal{M}^N(\xb^\prime_0)\in R) + \delta
}
holds for any two  $d$-adjacent initial values $\xb_0,\xb^\prime_0\in\mathbb{R}^n$.
\end{definition}

\medskip
In Definition \ref{Definition:DP},  inequality (\ref{eq:DPiv}) indicates that the system can plausibly deny any guess  from the eavesdroppers having $N$ output trajectories, using any value from its $d$-adjacency.
Namely, the system preserving the  differential initial-value privacy satisfies the  requirement (R2).

\medskip
\begin{remark}
The requirement (R1) indeed can be understood from a perspective of denability. Namely,
\begin{itemize}
  \item[(R1$^\prime$)] {\bf  [Deniability from non-identifiability]} any inference $\hat\xb_0$ about the true initial value from the eavesdroppers having an infinite number of output trajectories, can be denied by supplying any other value $\xb_0^\prime$ satisfying
\beeq{\label{eq:remark2}
\pdf\left(\mathbf{Y}_T| \hat\xb_0\right) = \pdf\left(\mathbf{Y}_T | \xb^\prime_0\right)\,.
}
\end{itemize}

\noindent
Note that the derivation of such $\xb_0^\prime$ needs extra computation such that (\ref{eq:remark2}) is fulfilled and the resulting $\xb_0^\prime$ may be very close to or far away from the inference $\hat\xb_0$. In contrast, the $\xb_0^\prime$ used to deny  $\hat\xb_0$ in (R2) is arbitrarily selected within a range to $\hat\xb_0$. In view of this, the plausible deniability in (R2) provides the system with a more convenient denial mechanism. On the other hand, it can be seen that (\ref{eq:remark2}) indicates that (\ref{eq:DPiv}) holds with $(\epsilon,\delta)=(0,0)$, yielding  that the eavesdroppers cannot  distinguish between $\hat\xb_0$ and $\xb_0^\prime$  with probability one.
The above analysis thus demonstrates that  intrinsic initial-value privacy and differential initial-value privacy are not inclusive mutually.

\end{remark}

\subsection{An Illustrative Example}

In this subsection, an illustrative example is presented to demonstrate the relation and practical difference between the above two types of privacy metrics.

\medskip
\noindent{\bf Example 1}. Consider system (\ref{eq:ini-sys}) with $\Ab = \begin{bmatrix}
        0 & 1 \\
        0 & -1
      \end{bmatrix}
$ and let $T=1$ and the private initial states $\xb_0=[x_{1,0};x_{2,0}]=[2;1]$. Let $\nub_t\backsim\mathcal{N}(0,\sigma_1^2\Ib_2)$ and $\omegab_t\backsim\mathcal{N}(0,\sigma_2^2)$.

\begin{itemize}
  \item[(a)] Let output $\yb_t=\Cb_1\xb_t$ with $\Cb_1 = \begin{bmatrix}
        1 & 1
      \end{bmatrix}$, and $\sigma_2=0$. We then obtain that
      \[\ba{l}
      \yb_0 = x_{1,0} + x_{2,0} \,,\\ \yb_1 = \begin{bmatrix}
        1 & 1
      \end{bmatrix}\nub_0 \,.
      \ea\]
\end{itemize}
For system with the output in Case (a), it is clear that $\pdf\left(\mathbf{Y}_T| \xb_0\right) = \pdf\left(\mathbf{Y}_T | \xb^\prime_0\right)$ for all $\xb_0^\prime=[x_{1,0}^\prime;x_{2,0}^\prime]$ satisfying $x_{1,0}^\prime+x_{2,0}^\prime=x_{1,0}+x_{2,0}$. This indicates that intrinsic initial-value privacy is preserved.

Regarding the differential privacy of initial values, we observe that $\yb_0$ is deterministic. By choosing any $\xb_0^\prime$ from $1$-adjacency of $\xb_0$ such that $x_{1,0}^\prime+x_{2,0}^\prime\neq x_{1,0}+x_{2,0}$, the inequality (\ref{eq:DPiv}) is not satisfied for any $\epsilon>0, 0.5>\delta>0$, which implies that the system does not preserve the differential privacy.


\begin{itemize}
  \item[(b)] Let output $\yb_t=\Cb_2\xb_t$ with $\Cb_2 = \begin{bmatrix}
        1 & 0
      \end{bmatrix}$, and  $\sigma_1=\sigma_2=1$. We then obtain that
      \[\ba{l}
      \yb_0 = x_{1,0} + \omegab_0\,,\quad \\ \yb_1 = x_{2,0} + \begin{bmatrix}
        1 & 0
      \end{bmatrix}\nub_0 + \omegab_1\,.
      \ea\]
\end{itemize}
If an infinite number of output trajectories were eavesdropped, then the eavesdroppers could obtain the expectations of $\yb_0,\yb_1$, denoted by $\mathbb{E}(\yb_0),\mathbb{E}(\yb_1)$. The initial values $x_{1,0},x_{2,0}$ then can be uniquely recovered by $x_{1,0}=\mathbb{E}(\yb_0)$ and $x_{2,0}=\mathbb{E}(\yb_1)$, leading to loss of intrinsic initial-value privacy. Though it is impossible to eavesdrop an infinite number of output trajectories, we note that the eavesdroppers may  obtain a large number of output trajectories, by which the initial values can be accurately inferred. To address this issue, we realize the system for $10^4$ times and use the MLE method. The resulting estimate of $\xb_0$ is $[2.0038;0.9920]$, leading to loss of the intrinsic initial-value privacy as well. In view of this, the privacy requirement (R1) and Definition \ref{Definition:IP} are of practical significance.

Regarding the differential privacy, we let $\yb_1=2,\yb_2=2$ be the eavesdropped output trajectory, and $\hat \xb_0=[1.5; 1.8]$ be a guess from the eavesdroppers. The system then denies this guess by stating that for example $\|\xb_0-\hat \xb_0\|> 0.5$. To verify whether this deny is plausible, the eavesdroppers thus compute $\pdf(\Yb_T|\xb_0)$ and $\pdf(\Yb_T|\xb_0')$ with any value $\xb_0'=[1; 1.2]$ satisfying $\|\xb_0'-\hat \xb_0\|> 0.5$, and find $\pdf(\Yb_T|\xb_0)\leq e^{0.6} \pdf(\Yb_T|\xb_0')$. In this way, the system has gained plausible deniability measured by $\delta =0.48$ and $\epsilon =0.6$. On the other hand, we randomly choose four adjacent initial values $\bar\xb_0=[1.4;1.7], \bar\xb_0^\prime=[1.6;1.8], \hat\xb_0=[1.5; 1.9], \hat\xb_0^\prime=[1.3;2]$. The resulting distributions of
$\yb_0,\yb_1$ under these initial values $\bar\xb_0, \bar\xb_0^\prime, \hat\xb_0, \hat\xb_0^\prime$ are presented in Figures \ref{fig_dp_a} and \ref{fig_dp_b}. From Figures \ref{fig_dp_a} and \ref{fig_dp_b}, it can be seen that the resulting probabilities of output trajectory $\Yb_T$ at any set are similar. This thus can imply that the system in Case (b) preserves the differential privacy of initial values.

\begin{figure}[H]
	\hspace*{0cm}
	\vspace*{0cm}
	\centering
	\includegraphics[width=6.5cm]{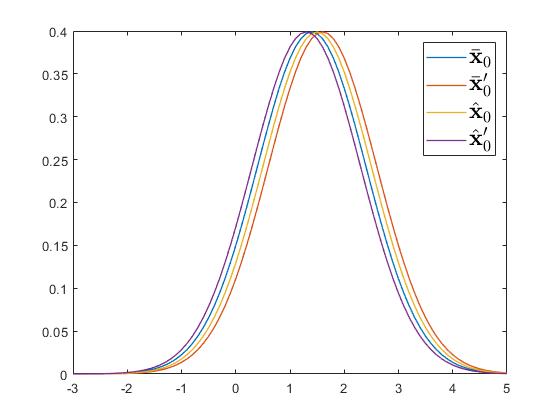}
	\caption{Distributions of $\yb_0$ under different initial values $\bar{\xb}_0, \bar{\xb}_0^\prime, \hat{\xb}_0, \hat{\xb}_0^\prime$.}
	\label{fig_dp_a}
\end{figure}

\begin{figure}[H]
	\hspace*{0cm}
	\vspace*{0cm}
	\centering
	\includegraphics[width=6.5cm]{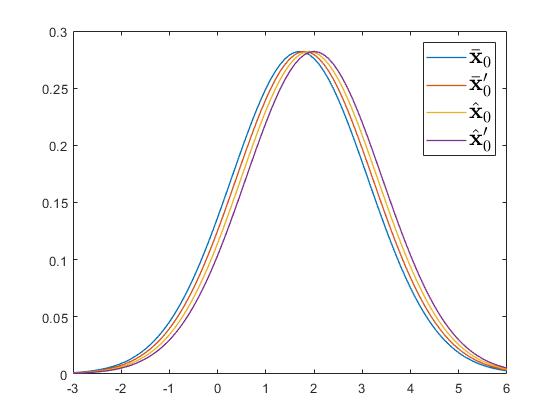}
	\caption{Distributions of $\yb_1$ under different initial values $\bar{\xb}_0, \bar{\xb}_0^\prime, \hat{\xb}_0, \hat{\xb}_0^\prime$.}
	\label{fig_dp_b}
\end{figure}

In view of the previous analysis for Cases (a) and (b), we can see that neither the intrinsic initial-value privacy implies the differential initial-value privacy, nor the differential initial-value privacy  implies the intrinsic initial-value privacy.
Therefore, both privacy metrics in Definitions 1 and 2 are mutually neither inclusive nor exclusive. $\blacksquare$

\subsection{Related Works}

Of relevance to our paper are recent works \cite{Ny-Pappas-TAC-2014,Hale-LQG,Kawano-Cao-19}. In \cite{Ny-Pappas-TAC-2014}, the considered linear dynamical systems take the form
\[\ba{rcl}
\xb_{t+1} &=& \Ab\, \xb_t + \Bb\omegab_t\,\\
\yb_{t} &=& \Cb\,\xb_t + \Db\omegab_t\,\\
\zb_t &=& \Lb\, \xb_t
\ea\]
where $\zb_t$ is a signal to be estimated, and the problem is to design a filter
\beeq{\label{eq:filter}\ba{rcl}
\hat \xb_{t+1} &=&  \Fb\, \hat \xb_t +  \Gb\, \yb_t\,\\
\hat \zb_t &=&  \Hb\, \hat \xb_t +  \Kb\, \yb_t + \nub_t
\ea}
with filter output $\hat \zb_t$, for the purpose of obtaining an optimal mean squared error between $\zb_t$ and $\hat \zb_t$, while preserving the differential privacy of system state trajectory $\Xb_T=(x_t)_{t=0}^T$, with the initial state considered in this paper as a particular case. Despite of this, we note that our framework can be extended to the scenario where $\Xb_T$ are sensitive, which will be explicitly addressed in Remark \ref{remark-rajectory-pri}.
On the other hand, in \cite{Ny-Pappas-TAC-2014}, the eavesdropped information is a filter output trajectory $\widehat{\Zb}_T=(\hat\zb_t)_{t=0}^T$, that is different from our paper where the eavesdroppers can directly measure system outputs.
To address the desired privacy concerns, \cite{Ny-Pappas-TAC-2014} studies the differential privacy of a composite mapping $\widehat{\Zb}_T=\mathcal{M}_{f}(\Xb_T):=\mathcal{M}_{2}\circ\mathcal{M}_{1}(\Xb_T)$, the mapping $\Yb_T = (\Ib_{T+1}\otimes C)\Xb_T + \Wb_T$ and the mapping $\widehat\Zb_T = \mathcal{M}_{2}(\Yb_T)$ is defined over the filter dynamics (\ref{eq:filter}). This is different mechanism compared to our mapping (\ref{eq:M^N}).

%

In \cite{Hale-LQG} a cloud-based linear quadratic regulation problem is studied for systems of the form
\[\ba{rcl}
\xb_{t+1} &=& \Ab\, \xb_t + \Bb\, \ub_t + \nub_t\,\\
\yb_{t} &=& \Cb\,\xb_t + \omegab_t\,\\
\ea\]
where the outputs $\yb_{t}$ are transmitted to the cloud for an optimal control, having the form
\[\ba{rcl}
\bar \xb_{t+1}&=& \Fb\, \bar \xb_{t} + \Gb\, \yb_t\,\\
\ub_t &=& \Hb\, \bar \xb_t\,.
\ea\]
The privacy issue of interest is to protect privacy of the state trajectory $\Xb_T$, against the eavesdroppers having access to the transmitted information consisting of $\yb_t,\ub_t$ by cyber attack.  To address the state-trajectory privacy, \cite{Hale-LQG} studies the differential privacy of the mapping  $\Yb_T =\mathcal{M}_{y}(\Xb_T):= (\Ib_{T+1}\otimes C)\Xb_T + \Wb_T$. Therefore, again this is different mechanism compared to our mapping (\ref{eq:M^N}).

In \cite{Kawano-Cao-19}, the authors consider a control system
\[\ba{rcl}
\xb_{t+1} &=& \Ab\, \xb_t + \Bb\,(\ub_t+\nub_t)\,\\
\yb_t &=& \Cb\,\xb_t +  \Db\,(\ub_t+\nub_t) + \omegab_t\,,
\ea\]
where $\nub_t,\omegab_t$ are injected input and measurement noise for privacy protection, and to achieve the desired control objective, the control input $\ub_t$ is designed as a form
\beeq{\label{eq:control}\ba{rcl}
\bar \xb_{t+1}&=& \Fb\, \bar \xb_{t} + \Gb\, \yb_t\,\\
\ub_t &=& \Hb\, \bar \xb_t + \Lb\, \yb_t\,.
\ea}
The eavesdroppers having access to an output trajectory $\Yb_T$ through cyber attack attempt to infer control inputs $\ub_t$ and initial values $\xb_0$. In this setting, \cite{Kawano-Cao-19} studies the differential privacy of a mapping from $\big((\ub_t)_{t=0}^T,x_0\big)$ to  $\Yb_T$ over the $\xb_t$-dynamics. Compared to our mapping (\ref{eq:M^N}), the considered mapping in \cite{Kawano-Cao-19} is similar in the sense that both are established over the $\xb_t$-dynamics, while is different  in the sense that there are $N$ output trajectories and no control input in our mapping (\ref{eq:M^N}).

%

In addition to the previous distinctions, we note that this paper also studies privacy issues when an infinite number of output trajectories are eavesdropped, while it is absent in \cite{Ny-Pappas-TAC-2014,Hale-LQG,Kawano-Cao-19}. This extra study benefits us to understand whether the initial-value privacy can be preserved if the eavesdroppers have obtained a large number of output trajectories.

\section{Initial-Value Privacy of General Linear Systems}


In this section, both intrinsic and differential initial-value privacy of  systems (\ref{eq:ini-sys}) are analyzed. We first present the following result on the equivalence of intrinsic initial-value privacy and unobservability.

\medskip
\begin{proposition}\label{Proposition-IP}
The system (\ref{eq:ini-sys}) preserves intrinsic initial-value privacy \emph{if and only if} $(\Ab,\Cb)$ is not observable, i.e., $\rank(\mathbf{O}_{ob})<n$.
\end{proposition}

\medskip
\begin{remark}
\label{Remark:q}
Observability has been extensively studied in the fields of estimation \cite{Clanents-1983} and feedback control \cite{Sontag}. In \cite{Kawano-Cao-19}, for a linear control system the input observability is also explored to preserve differential privacy of control inputs and initial states. In Proposition \ref{Proposition-IP}, the intrinsic initial-value privacy and observability are bridged for linear systems (\ref{eq:ini-sys}).

\end{remark}
\medskip

Next, the differential privacy of initial values for (\ref{eq:ini-sys}) is studied.
As in \cite{Ny-Pappas-TAC-2014}, we define $\mathcal{Q}(w):=\frac{1}{\sqrt{2\pi}}\int_{w}^\infty e^{-\frac{v^2}{2}}d v$, and  $\kappa (\epsilon,\delta):=\frac{\mathcal{Q}^{-1}(\delta)+\sqrt{(\mathcal{Q}^{-1}(\delta))^2+2\epsilon}}{2\epsilon}$.

\medskip
\begin{theorem}\label{Theorem-DP}
Suppose that $(\mathbf{V}_T;\mathbf{W}_T)$ are random variables according to  $(\mathbf{V}_T;\mathbf{W}_T)\backsim\mathcal{N}(0,\Sigma_{T})$.
Then  the dynamical system (\ref{eq:ini-sys}) preserves $(\epsilon,\delta)$-differential privacy of initial state under $d$-adjacency, with $\epsilon>0$ and $0.5>\delta>0$, if
\beeq{\ba{l}\label{eq:sigma}
\vspace{1mm}
\sigma_m\left(\begin{bmatrix}
        \mathbf{H}_T & \Ib_{m(T+1)}
         \end{bmatrix} \Sigma_T
         \begin{bmatrix}
        \mathbf{H}_T & \Ib_{m(T+1)}
         \end{bmatrix}^\top\right)  \geq d^2N\|\mathbf{O}_T\|^2\kappa (\epsilon,\delta)^2\,.
\ea}
\end{theorem}
\medskip

\begin{remark}
In Theorem \ref{Theorem-DP}, $\nub_t,\omegab_t$ are assumed to admit Gaussian distributions. This renders the mapping (\ref{eq:M^N}) to be a Gaussian mechanism \cite{Ny-Pappas-TAC-2014,Dwork-2}, resulting in the $(\epsilon,\delta)$-differential privacy. One may wonder if assuming  Laplacian noise $\nub_t,\omegab_t$  would lead to a stronger $(\epsilon,0)$-differential privacy, as in \cite{cortes2018tcns}. However, we note that the resulting mapping (\ref{eq:M^N}) is not a Laplace mechanism, because there is no guarantee that $\mathbf{H}_T \mathbf{V}_T + \mathbf{W}_T$ is still Laplacian, even if $\nub_t,\omegab_t$ are Laplacian variables.
\end{remark}

\medskip
\begin{remark}\label{remark-rajectory-pri}
Though only initial values of system (\ref{eq:ini-sys}) are treated as private information, we note that the results in Theorem \ref{Theorem-DP} can be extended to the case that all states $\xb_t$ are sensitive.  For dynamical systems (\ref{eq:ini-sys}), the outputs $\yb_k$ for all $k\geq t$ in $\Yb_{T}$ contain the information of $\xb_t$, rendering a mapping  from $\xb_t$ and $\Yb_{T-t}$ as
\[
\Yb_{T-t} = \mathcal{M}_t(\xb_t) := \mathbf{O}_{T-t}\xb_t + \Hb_{T-t} \Vb_{T-t} + \Wb_{T-t}\,.
\]
By combining all $\mathcal{M}_t(\xb_t)$, $t=0,1,\ldots,T$ together, one then can  establish a mapping from the state trajectory $(\xb_t)_{t=0}^T$ to the output trajectory $\Yb_T$. For such a combined mapping, following the arguments in the proof of Theorem \ref{Theorem-DP}, one then can establish $(\epsilon,\delta)$-differential privacy of the state trajectory $(\xb_t)_{t=0}^T$ with some $\epsilon>0$ and $0.5>\delta>0$. In this way, our framework can be further applied to solve the problems in \cite{Ny-Pappas-TAC-2014,Hale-LQG}, where the state trajectory is private information.

\end{remark}
\medskip

We note that given any covariance matrix $\Sigma_T>0$, there always exist $\epsilon>0$ and $0.5>\delta>0$, depending on the norm of extended observability matrix $\mathbf{O}_{T}$ such that (\ref{eq:sigma}) is satisfied, yielding the $(\epsilon,\delta)$-differential initial-value privacy. Thus, the $(\epsilon,\delta)$-differential initial-value privacy and the intrinsic initial-value privacy are mutually independent, with the latter determined  by the unobservability of systems (\ref{eq:ini-sys}), i.e., $\rank(\mathbf{O}_{ob})<n$ by Proposition \ref{Proposition-IP}.
If the noise can be designed, then there always exists a sufficiently large covariance matrix $\Sigma_{T}$ such that (\ref{eq:sigma}) holds for any privacy budgets $\varepsilon>0, 0.5>\delta>0$. To have a better view of this, we consider a particular case that $\nub_t$ and $\omegab_t$ are i.i.d. random variables. The following corollary can be easily derived by verifying the condition (\ref{eq:sigma}).

\medskip
\begin{corollary}
Suppose $\nub_t$ and $\omegab_t$, $t=0,1,\ldots,T$ are i.i.d. random variables according to $\nub_t\backsim\mathcal{N}(0,\sigma_{\nub}^2\Ib_n)$ and $\omegab_t\backsim\mathcal{N}(0,\sigma_{\omegab}^2\Ib_m)$. Then for any $\epsilon>0$, $0.5>\delta>0$, and all $\sigma_{\nub}\geq0$ and $\sigma_{\omegab} \geq d\sqrt{N}\|\mathbf{O}_T\|\kappa (\epsilon,\delta)$,
the dynamical system (\ref{eq:ini-sys}) preserves $(\epsilon,\delta)$-differential privacy of initial state under $d$-adjacency.
\end{corollary}

\medskip
\begin{remark}
  Though in Corollary 1 arbitrary $(\epsilon,\delta)$-differential privacy can be achieved by choosing a sufficiently large $\sigma_{\omegab}$, this doesn't mean that the process noise $\nub_t$ does not contribute to the differential privacy. In fact, simple calculations following the proof of Theorem \ref{Theorem-DP} can lead to a less restrictive condition as
  \[
  \|\mathbf{O}_T^\top (\sigma_{\nub}^2\mathbf{H}_T\mathbf{H}_T^\top + \sigma_{\omegab}^2\Ib_m)^{-1}\mathbf{O}_T\| \leq \frac{1}{d^2N\kappa (\epsilon,\delta)^2}\,,
  \]
  from which it can be seen that $\sigma_{\nub}$ also plays a role in achieving arbitrary differential privacy of initial values.
\end{remark}

\medskip
{
\begin{remark}
If the considered systems take a time-varying form
\beeq{\label{eq:ini-ltvsys}\ba{rcl}
\xb_{t+1} &=& \Ab_t\, \xb_t + \nub_t\,\\
\yb_{t} &=& \Cb_t\,\xb_t + \omegab_t\,
\ea}
where the state and output matrices $\Ab_t,\Cb_t$ vary as $t$ evolves, it can be easily verified that the claims in Proposition \ref{Proposition-IP} and Theorem \ref{Theorem-DP} are still preserved by replacing the  observability matrix $\mathbf{O}_{ob}$ and the extended observability matrix $\mathbf{O}_{T}$ by their time-varying version $\widehat{\mathbf{O}}:= \big[
                 \Cb_0^\top,\,
                 \Ab_0^\top \Cb_1^\top,\,
                 \cdots ,\,
                 \Ab_0^\top\cdots\Ab_{T-1}^\top \Cb_T^\top
               \big]^\top$.
\end{remark}
}

\section{Intrinsic Initial-Value Privacy of Networked Linear Systems}
\label{sec-4}

The system (\ref{eq:ini-sys}) can also be understood from a network system perspective, e.g.,\cite{Mesbahi-Egerstedt-10}. Let $x_{i,t}$ be the $i$-th entry of $\xb_t$. If each $x_{i,t}$ is viewed as the dynamical state of a node,  the matrix $\Ab$ would indicate a graph of interactions among the  nodes. If each entry of $\yb_t$ is viewed as the measurement of a sensor, then the matrix $\Cb$ would indicate a graph of interactions between the nodes and the sensors.

In view of this,  we consider a network consisting of $n$ network nodes and $m$ sensing nodes, leading to a network node set $\V=\{1,\ldots,n\}$ and a sensing node set $\mathrm{V}_{\textnormal S}=\{s_1,\ldots,s_m\}$ \footnote{To be distinguished with notations for nodes in the interaction graph $\G $, we use $s_i$ to denote the $i$-th sensing node whose measurement is $y_i$.}, respectively. Define the interaction graph $\G=(\V,\E)$ with  edge set $\E \subset \V \times \V $, and the sensing graph $\G_{\textnormal S}=(\V ,\mathrm{V}_{\textnormal S},\E_{\textnormal S})$ with edge set $\E_{\textnormal S}\subset \V\times \V_{\textnormal S} $.
Let $\Ab=[a_{ij}]\in\R^{n\times n}$ and $\Cb=[c_{ij}]\in\R^{m\times n}$.

To this end, this section aims to study how topological effects affect the privacy analysis of the networked system (\ref{eq:ini-sys}) with $(\Ab,\Cb)$ being a configuration complying with the graphs $\G, \G_{\textnormal S}$, i.e., if $a_{ij}=0$ for $(j,i)\notin\E $ and  $c_{ij}=0$ for $(j,s_i)\notin\E_{\textnormal S}$.

\medskip
\begin{remark}
For the LTV systems (\ref{eq:ini-ltvsys}), the corresponding network structure becomes time-varying,  where matrix $\Ab_t$ indicates a graph of interactions among the network nodes at time $t$ and matrix $\Cb_t$ indicates a graph of interactions between the network and sensing nodes at time $t$.
Thus we define the time-varying interaction graph $\G_t=(\V,\E_t)$ with  edge set $\E_t \subset \V \times \V $, and the sensing graph $\G_{\textnormal S,t}=(\V ,\mathrm{V}_{\textnormal S},\E_{\textnormal S,t})$ with edge set $\E_{\textnormal S,t}\subset \V\times \V_{\textnormal S}$.
\end{remark}

\subsection{Intrinsic privacy of individual initial values}

It is noted that in Definition 1 regarding intrinsic initial-value privacy, the initial state vector $\xb_0$ is considered as a whole, and we suppose that the eavesdroppers have no prior knowledge of any individual initial values. In the following, we  present several definitions that refine the notion in Definition 1 to dynamical networked systems (\ref{eq:ini-sys}) by studying intrinsic privacy of individual initial values, i.e., $x_{i,0}$, against eavesdroppers having knowledge of the whole sensor measurements (i.e., $\yb_t$) and  initial values of some network nodes.  For convenience, we term the set of nodes whose initial values are prior knowledge to eavesdroppers as a \emph{public disclosure set}.

\medskip
\begin{definition}\label{Definition:LIPP}
For any given configuration $(\Ab, \Cb)$ complying with graphs $\G,\G_{\textnormal S}$, take $i\in \V $ and let ${\mathrm{P}}\subset \V $. The networked system (\ref{eq:ini-sys}) preserves \emph{intrinsic initial-value privacy of node $i$} w.r.t. public disclosure set ${\mathrm{P}}$ if for any initial state $\xb_0\in\mathbb{R}^n$, there exists an $\xb^\prime_0=[x_{1,0}^\prime;\ldots;x_{n,0}^\prime]\in\mathbb{R}^n$ such that $x_{i,0}\neq x_{i,0}^\prime$, $x_{j,0}= x_{j,0}^\prime$ for all $j\in {\mathrm{P}}$, and
\beeq{\label{eq:RIPiv}
\pdf\left(\mathbf{Y}_T| \xb_0\right) = \pdf\left(\mathbf{Y}_T | \xb^\prime_0\right)\,.
}
\end{definition}

\medskip
\begin{remark}
  The equality (\ref{eq:RIPiv}) indicates that even if the initial values of some nodes $j\in{\mathrm{P}}$ are public, the initial value $x_{i,0}$ cannot be identified from trajectories of $\yb_t$, even with an infinite number of realizations of the dynamic networked system (\ref{eq:ini-sys}).
\end{remark}
\medskip

\begin{remark}
  If the eavesdroppers have no prior knowledge of any node initial states, the above definition is also applicable with ${\mathrm{P}}=\emptyset$. In this case, according to Proposition \ref{Proposition-IP}, the notion of intrinsic initial-value privacy of node $i$ is related to state variable unobservability of state $x_{i,t}$, that is a dual notion of state variable uncontrollability in \cite{Blackhall2010}.
\end{remark}
\medskip

Let $l=|{\mathrm{P}}|$ and ${\mathrm{P}}=\{p_1,\ldots,p_l\}\subset\V $. Define $\bar{{\mathrm{P}}}:=\{\bar p_1,\ldots,\bar p_{n-l}\}=\V \backslash{\mathrm{P}}$.
For convenience, we further let  $\Eb_{\mathrm{P}}=[\eb_{p_1},\ldots,\eb_{p_l}]\in\R^{n\times l}$ and $\Eb_{\bar{{\mathrm{P}}}}=[\eb_{\bar p_1},\ldots,\eb_{\bar p_{n-l}}]\in\R^{n\times (n-l)}$, and $\Kb_{j}^{ob}$ be the $j$-the column of matrix $\mathbf{O}_{ob}$.

\medskip
\begin{theorem}\label{Theorem-LIPP}
 Let the dynamical networked system (\ref{eq:ini-sys}) be equipped with configuration $(\Ab, \Cb)$ complying with graphs $\G ,\G_{\textnormal S}$.  Let $i\in \V $ and ${\mathrm{P}}\subset \V $ with $i\notin {\mathrm{P}}$. The following statements are equivalent.
  \begin{itemize}
    \item[a).] The networked system (\ref{eq:ini-sys}) preserves \emph{intrinsic initial-value privacy of  node $i$} w.r.t. ${\mathrm{P}}$.
        \vspace{1mm}
    \item[b).] $
  \rank(\mathbf{O}_{ob}\Eb_{\bar {{\mathrm{P}}} }) = \rank([\Kb_{i_1}^{\textnormal ob},\ldots,\Kb_{i_{n-l-1}}^{\textnormal ob}])
  $ with $\{i_1,\ldots,i_{n-l-1}\}=\V \backslash({\mathrm{P}}\cup \{i\})$.
  \vspace{1mm}
  \item[c).] $\rank\left(\begin{bmatrix} \mathbf{O}_{ob} \cr \Eb_{\mathrm{P}}^\top \cr \eb_i^\top  \end{bmatrix}\right) = \rank\left(\begin{bmatrix}\mathbf{O}_{ob} \cr \Eb_{\mathrm{P}}^\top\end{bmatrix} \right) + 1$.
  \end{itemize}
\end{theorem}

\medskip

\begin{remark}
Recalling (\ref{eq:M}), it can be seen that each initial state $x_{j,0}$ is multiplied by the $j$-th column of the (extended) observability matrix.
We term each $j$-th column of the observability matrix as a ``feature" vector of the corresponding node $j$, and the  columns corresponding to nodes in public set ${\mathrm{P}}$ and unpublic set $\bar{\mathrm{P}}$ as public and unpublic ``feature" vectors, respectively. Then, the equivalence of statements $a)$ and $b)$ demonstrates that the intrinsic initial-value privacy of node $i$ is preserved if and only if its ``feature" vector can be expressed by a linear combination of the remainder unpublic ``feature" vectors, i.e., the $i$-th ``feature" vector is encrypted by the remainder unpublic ones.
\end{remark}

\medskip
\begin{remark}\label{Remark-8}
In Theorem \ref{Theorem-LIPP}, the equivalence of statements $a)$ and $c)$ demonstrates that the intrinsic initial-value privacy of node $i$ is preserved if and only if $\eb_i^\top$ does not belong to the $\mathrm{P}$-\emph{extended} observable subspace, denoted by $\span\left(\begin{bmatrix}\mathbf{O}_{ob} \cr \Eb_{\mathrm{P}}^\top\end{bmatrix} \right)$.
\end{remark}

\medskip
\begin{remark}\label{Remark-9}
It is worth noting that the verification of statement $c)$ can be simplified as
  \begin{itemize}
    \item[c$^\prime$).] $\rank\left(\begin{bmatrix}  \mathbf{O}_{ob} \cr \eb_i^\top \end{bmatrix}\Eb_{\bar {\mathrm{P}}}\right) = \rank\left(\mathbf{O}_{ob}\Eb_{\bar {{\mathrm{P}}} }\right) + 1$\,.
  \end{itemize}
This can be easily verified by using the following facts that
   \[\ba{l}
   \vspace{1mm}
\rank\left(\begin{bmatrix} \mathbf{O}_{ob} \cr \Eb_{\mathrm{P}}^\top \cr \eb_i^\top  \end{bmatrix}\right) = \rank\left(\begin{bmatrix}  \mathbf{O}_{ob} \cr \eb_i^\top \end{bmatrix}\Eb_{\bar {{\mathrm{P}}} }\right) + \rank(\Eb_{ {{\mathrm{P}}} }) \,\\
\rank\left(\begin{bmatrix}\mathbf{O}_{ob} \cr \Eb_{\mathrm{P}}^\top\end{bmatrix} \right)=\rank\left(\mathbf{O}_{ob}\Eb_{\bar {{\mathrm{P}}} }\right) + \rank(\Eb_{{{\mathrm{P}}} })\,.
\ea\]
In view of this, we occasionally use c$^\prime$) to replace c) in Theorem \ref{Theorem-LIPP} in the sequel.
\end{remark}
\medskip
In Theorem \ref{Theorem-LIPP}, explicit rank conditions are proposed to determine whether the intrinsic initial-value privacy of individual nodes is preserved, with respect to any given public disclosure set ${\mathrm{P}}$.   On the other hand, for a networked system, one may naturally ask what is the maximum allowable disclosure such that there always exists at least one  node whose initial-value privacy is preserved. To address this issue, the network privacy index is introduced below.

\medskip
\begin{definition}
The networked system (\ref{eq:ini-sys}) achieves level-$l$ network privacy, if for any public disclosure set ${{\mathrm{P}}}\subset \V $ with $|{{\mathrm{P}}}|=l$, there exists a node $i\in \V \backslash {{\mathrm{P}}}$ whose intrinsic initial-value privacy is preserved w.r.t. ${{\mathrm{P}}}$. The network privacy index  of (\ref{eq:ini-sys}), denoted as $\bf I_{rp}$, is defined as the maximal value of $l$ such that level-$l$ relative privacy is achieved.
\end{definition}

\medskip
\begin{proposition}\label{Proposition-NPI}
   The network privacy index of networked system (\ref{eq:ini-sys}) is ${\mathbf I_{rp}} = n-\rank(\mathbf{O}_{ob})-1$.
\end{proposition}

\medskip
\begin{remark}
By Definition 4,  the full initial value is not  disclosed irrespective of which ${\mathbf I_{rp}}$ nodes are public. It is clear that a larger ${\mathbf I_{rp}}$ means a stronger privacy-preservation ability of the networked system (\ref{eq:ini-sys}).
 According to Proposition \ref{Proposition-NPI}, this further implies that a networked system possesses a better privacy-preservation ability, if the dimension of its unobservable subspace (i.e., $n-\rank(\mathbf{O}_{ob})$) is higher.
\end{remark}

\medskip
\noindent{\bf Example 2}. We present an example  to illustrate Theorem \ref{Theorem-LIPP} and Proposition \ref{Proposition-NPI}.
Consider a networked system (\ref{eq:ini-sys}) with $(\Ab,\Cb)$ complying with the graphs $\G,\G_{\textnormal S}$ in Fig. \ref{fig:network-1}, consisting of 20 network nodes and 2 sensing nodes. Each edge of $\E, \E_{\textnormal S}$ is assigned with the same weight 1.

\begin{figure}[H]
	\hspace*{0cm}
	\vspace*{0cm}
	\centering
	\includegraphics[width=7.5cm]{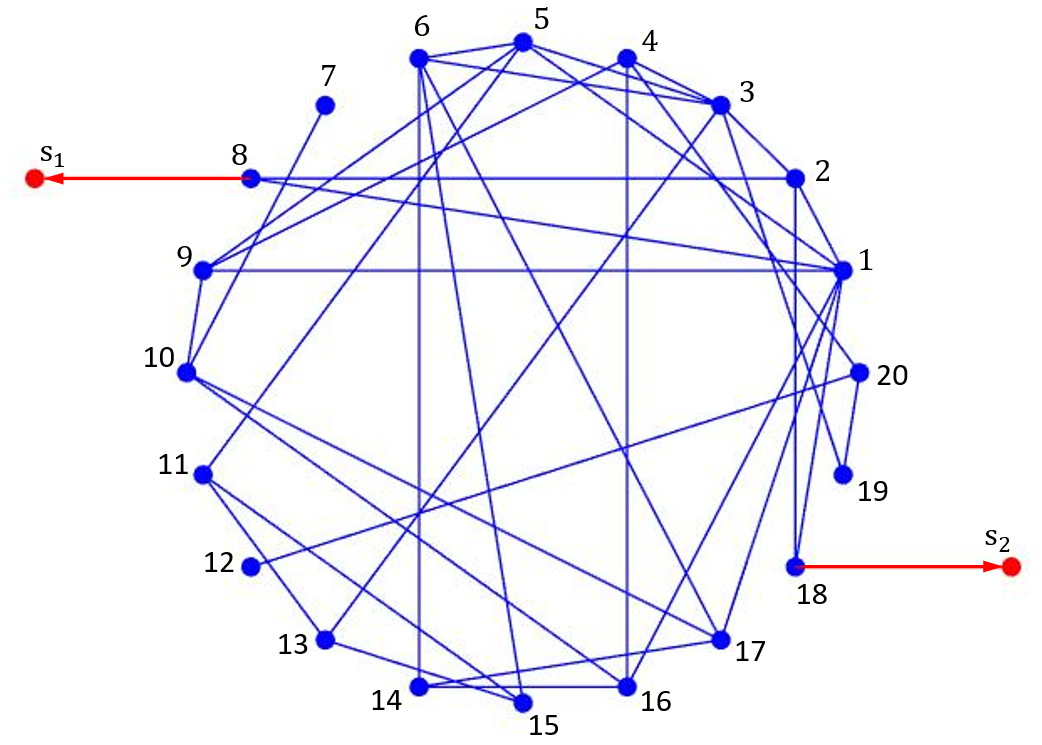}
	\caption{Network topologies $(\G,\G_{\textnormal S})$ with 20 network nodes (blue circles) and 2 sensing nodes (red circles). The edges in $\E$ and $\E_{\textnormal S}$ are drawn in blue and red lines, respectively, where the lines without arrows denote bidirectional edges.
}
	\label{fig:network-1}
\end{figure}
Firstly, suppose the eavesdroppers have no prior knowledge of  initial values of any nodes, i.e., ${\mathrm{P}}=\emptyset$. We observe that $\rank(\mathbf{O}_{ob})=17$. According to Proposition \ref{Proposition-NPI}, this indicates that the network privacy index  ${\mathbf I_{rp}} =2$. Moreover, for all $i\in \V\backslash\{8,18\}$, there holds
\[
\rank\left(\begin{bmatrix} \mathbf{O}_{ob} \cr \eb_i^\top  \end{bmatrix}\right) = 18\,.
\]
According to Theorem \ref{Theorem-LIPP}, this indicates that the networked system preserves intrinsic initial-value privacy of all $i\in \V\backslash\{8,18\}$.

Let the public disclosure set ${\mathrm{P}}=\{1,2,7\}$.  Taking into account the intrinsic initial-value privacy of individual nodes, we note that
\[
\rank\left(\begin{bmatrix} \mathbf{O}_{ob} \cr \Eb_{{\mathrm{P}}}\end{bmatrix}\right)=19\,,\quad
\rank\left(\begin{bmatrix} \mathbf{O}_{ob} \cr \Eb_{{\mathrm{P}}} \cr\eb_i^\top  \end{bmatrix}\right) = 20\,
\]
for all $i\in\{3,5,6,9,10,11,12,13,14,16,17\}$.
This verifies statement $c)$ and thus indicates that the networked system preserves intrinsic initial-value privacy of all nodes $i\in\{3,5,6,9,10,11,12,13,14,16,17\}$, even if the initial states of nodes $1,2,7$ are public. $\blacksquare$

\subsection{Generic intrinsic initial-value privacy}

In the previous subsection, the intrinsic initial-value privacy of individual nodes w.r.t.  the public disclosure set ${{\mathrm{P}}}$ of networked systems (\ref{eq:ini-sys}) is studied, and a network privacy index $\mathbf I_{rp}$ is proposed to quantify the privacy of networked system (\ref{eq:ini-sys}). In the following, we turn to study the effect of network structure $(\G ,\G_{\textnormal S})$ to the intrinsic privacy and the network privacy index. To be precise, we demonstrate that these properties are indeed generic, i.e., are fulfilled for almost all edge weights under any network structure $(\G ,\G_{\textnormal S})$.

\medskip
\begin{theorem}\label{Theorem-GIVP}
  Let ${{\mathrm{P}}}\subset \V $ and $i\in\V $. Then the intrinsic initial-value privacy of node $i$ w.r.t. ${{\mathrm{P}}}$ is generically determined by the network topology. To be precise, exactly one of the following statements holds for any non-trivial network structure $(\G ,\G_{\textnormal S})$.
  \begin{itemize}
    \item[(i)] The intrinsic initial-value privacy of node $i$ is preserved generically, i.e., for almost all configurations $(\Ab,\Cb)$ complying with the network structure $(\G ,\G_{\textnormal S})$.
    \item[(ii)] The intrinsic initial-value privacy of node $i$ is lost generically, i.e., for almost all configurations $(\Ab,\Cb)$ complying with the network structure $(\G ,\G_{\textnormal S})$.
  \end{itemize}  %
\end{theorem}
\medskip

{
Theorem \ref{Theorem-GIVP} demonstrates that given any  network structure $(\G ,\G_{\textnormal S})$ and ${{\mathrm{P}}}\subset \V $, the intrinsic initial-value privacy of node $i$ is either preserved or lost generically. We note that if there exists a configuration $(\Ab,\Cb)$ complying with $(\G ,\G_{\textnormal S})$ such that the intrinsic initial-value privacy of node $i$ is preserved (or lost), there is no guarantee that such property is preserved (or lost) generically. This is different from other common generic properties like structural controllability \cite{Lin-1974}, for which if there exists a configuration such that a linear system is  controllable, then it must be controllable for almost {all} configurations, i.e., structurally controllable.  To have a better view of this, the following examples are formulated.
\medskip

\noindent{\bf Example 3}.
Consider the intrinsic initial-value privacy of node 1  with $(\Ab,\Cb)$ complying with the network structure in Fig. \ref{fig:network-5}. Let $T=2$, and the system output trajectory $\Yb_T$ is given by (\ref{eq:M}) with
\[
\mathbf{O}_{T}=\begin{bmatrix}
                  c_{11} & 0 & c_{13} \\
                  0 & c_{11}a_{12} & 0 \\
                  c_{11}a_{12}a_{21} & 0 & c_{11}a_{12}a_{23}
                \end{bmatrix}\,.
\]
It is clear that $\mathbf{O}_{T}$ is full-rank for almost all configurations $(\Ab,\Cb)$ complying with  Fig. \ref{fig:network-5}. Thus, for any $\xb_0,\xb_0^\prime$ with $x_{1,0}\neq x_{1,0}^\prime$,  $\mathbf{O}_{T}\xb_0 \neq \mathbf{O}_{T}\xb_0^\prime$ and thus $\pdf\left(\mathbf{Y}_T| \xb_0\right) \neq \pdf\left(\mathbf{Y}_T | \xb^\prime_0\right)$ hold for almost all configurations $(\Ab,\Cb)$ complying with  Fig. \ref{fig:network-5}. According to Definition \ref{Definition:LIPP}, this indicates that the intrinsic initial-value privacy of node 1 is lost generically.

\begin{figure}[H]
	\hspace*{0cm}
	\vspace*{0cm}
	\centering
	\includegraphics[width=3cm]{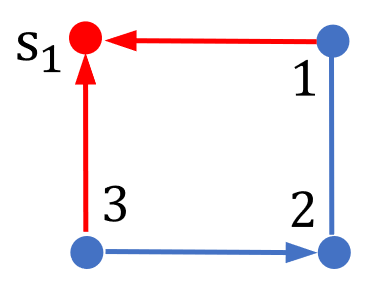}
	\caption{Network topologies $(\G,\G_{\textnormal S})$ with 3 network nodes (blue circles) and 1 sensing nodes (red circles). The line without arrow denotes a bidirectional edge.
}
	\label{fig:network-5}
\end{figure}

However, by letting the configuration $(\Ab,\Cb)$ be such that $c_{11}a_{23} = c_{13}a_{21}$, simple calculations show that
$
\pdf\left(\mathbf{Y}_T| \xb_0\right) = \pdf\left(\mathbf{Y}_T | \xb^\prime_0\right)
$
holds for any $\xb_0,\xb_0^\prime$ with $x_{1,0}\neq x_{1,0}^\prime$ and $c_{13}(x_{3,0}^\prime -x_{3,0}) = c_{11}(x_{1,0}-x_{1,0}^\prime)$. According to Definition \ref{Definition:LIPP}, this indicates that the intrinsic initial-value privacy of node 1 is preserved under the configuration $(\Ab,\Cb)$ satisfying $c_{11}a_{23} = c_{13}a_{21}$.

Thus, even if there exists a configuration such that the intrinsic initial-value privacy of node $i$ is preserved, this property may still be lost generically. This validates the statement (i) of Theorem \ref{Theorem-GIVP}. $\blacksquare$

\medskip
\noindent{\bf Example 4}. Consider the intrinsic initial-value privacy of node 4  under the network structure in Fig. \ref{fig:network-6}. Let $T=3$, and the  output trajectory $\Yb_T$ is given by (\ref{eq:M}) with
\[
\mathbf{O}_{T}=\begin{bmatrix}
                  c_{11} & 0 & c_{13} &0\\
                  c_{11}a_{11} & c_{11}a_{12} & 0 &c_{11}a_{14}\\
                  c_{11}a_{11}^2 & c_{11}a_{12}(a_{11}+a_{22}) & c_{11}a_{12}a_{23} & c_{11}a_{11}a_{14}\\
                  c_{11}a_{11}^3 & \ast & ? & c_{11}a_{11}^2a_{14}\\
                \end{bmatrix}\,
\]
with $\ast=c_{11}a_{12}(a_{11}^2+a_{11}a_{22}+a_{22}^2)$ and $?=c_{11}a_{12}a_{23}(a_{11}+a_{22})$.
Simple calculations then can show that $\pdf\left(\mathbf{Y}_T| \xb_0\right) = \pdf\left(\mathbf{Y}_T | \xb^\prime_0\right)$ holds for any $\xb_0,\xb_0^\prime$ with $x_{4,0}\neq x_{4,0}^\prime$ and  $x_{j,0}^\prime= x_{j,0}+\eta_j$ for $j=1,2,3$, where $\eta_j$'s satisfy
\beeq{\label{eq:MEqs}\ba{l}
c_{11}\eta_1 + c_{13} \eta_3 = 0\,\\
a_{11}\eta_1 + a_{12}\eta_2 = a_{14}(x_{4,0}- x_{4,0}^\prime)\,\\
a_{22}\eta_2 + a_{23}\eta_3 =0\,.
\ea}
It is clear that the above matrix equations (\ref{eq:MEqs}) have a solution for almost all configurations $(\Ab,\Cb)$ complying with  Fig. \ref{fig:network-6}, which, by Definition \ref{Definition:LIPP}, indicates that the  intrinsic initial-value privacy of node 4 is preserved generically.

\begin{figure}[H]
	\hspace*{0cm}
	\vspace*{0cm}
	\centering
	\includegraphics[width=5cm]{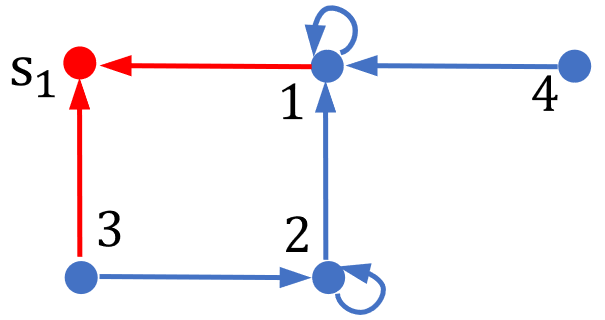}
	\caption{Network topologies $(\G,\G_{\textnormal S})$ with 4 network nodes (blue circles) and 1 sensing nodes (red circles).
}
	\label{fig:network-6}
\end{figure}
However, by letting the configuration $(\Ab,\Cb)$  be such that $c_{13}a_{11}a_{22}+c_{11}a_{12}a_{23}=0$ and $c_{11}a_{14}a_{22}\neq0$, it can be seen that there exists no $\eta_j$'s such that the matrix equations (\ref{eq:MEqs}) holds for any $x_{4,0}\neq x_{4,0}^\prime$, and
\[
(a_{22}+a_{11})\yb_1-a_{11}a_{22}\yb_0-\yb_2 = c_{11}a_{14}a_{22}x_{4,0} + g(\Wb_T,\Vb_T)
\]
with $g(\Wb_T,\Vb_T)$ being some function of noise vectors $\Wb_T,\Vb_T$.
This  immediately yields $\pdf\left(\mathbf{Y}_T| \xb_0\right) \neq \pdf\left(\mathbf{Y}_T | \xb^\prime_0\right)$  for all $\xb_0,\xb_0^\prime$ with $x_{4,0}\neq x_{4,0}^\prime$.
By Definition \ref{Definition:LIPP}, this indicates that  the intrinsic initial-value privacy of node 4 is lost under the configuration $(\Ab,\Cb)$ satisfying $c_{13}a_{11}a_{22}+c_{11}a_{12}a_{23}=0$ and $c_{11}a_{14}a_{22}\neq0$.

Thus, even if there exists a configuration such that the intrinsic initial-value privacy of node $i$ is lost, such property may still be preserved generically. This is consistent with  the statement (ii) of Theorem \ref{Theorem-GIVP}. $\blacksquare$
}

\medskip

%

Denote the maximal rank of $\mathbf{O}_{ob}\Eb_{\bar {{\mathrm{P}}} }$ over the matrix pairs $(\Ab,\Cb)$ that comply with the network structure $(\G ,\G_{\textnormal S})$ as $n_{\rm P}^{ob}$. By Lemma \ref{lemma-1} in Appendix \ref{app-theo-3}, it is noted that $\rank(\mathbf{O}_{ob}\Eb_{\bar {{\mathrm{P}}} })=n_{\rm P}^{ob}$ holds for almost all configurations $(\Ab,\Cb)$ complying with the network structure $(\G ,\G_{\textnormal S})$.  We  now present a practical verification approach of the intrinsic initial-value privacy of a node $i$ over the network structure $(\G ,\G_{\textnormal S})$.

Let the entries of $\mathbf{A}$ and $\mathbf{C}$ be generated independently and randomly according to the uniform distribution over the interval $[0,1]$, and complying with the structure  $(\G ,\G_{\textnormal S})$. We introduce the following three conditions:
\begin{itemize}
  \item[($C1$)] $\rank\left(\begin{bmatrix}  \mathbf{O}_{ob} \cr \eb_i^\top \end{bmatrix}\Eb_{\bar {{\mathrm{P}}} }\right) = n_{\rm P}^{ob} + 1$.
      \vspace{1mm}
   \item[($C2$)]  $\rank\left([\Kb_{i_1}^{\textnormal ob},\Kb_{i_2}^{\textnormal ob},\ldots,\Kb_{i_{n-l-1}}^{\textnormal ob}]\right) = n_{\rm P}^{ob}$ with $\{i_1,i_2,\ldots,i_{n-l-1}\}=\V \backslash({\mathrm{P}}\cup \{i\})$.
       \vspace{1mm}
   \item[($C3$)]  $\rank\left(\begin{bmatrix}  \mathbf{O}_{ob} \cr \Eb_{\mathrm{P}}^\top  \cr \eb_i^\top \end{bmatrix}\right) = n_{\rm P}^{ob} + l + 1.$
 \end{itemize}

\vspace{1mm}
 Denote $\mathsf{A},\mathsf{C}$ as one realized  sample from this randomization, from which we define the following event
$$
\mathcal{E}_{(\mathsf{A}, \mathsf{C})}:=\{(\mathsf{A}, \mathsf{C}): \mbox{either}\,(C1),\,(C2)\,or\,(C3)\ {\rm holds} \}.
$$
We present the following result.

\medskip
\begin{theorem}\label{Theo-verify}
The following statements hold.
\begin{itemize}
  \item[(i)] If the event  $\mathcal{E}_{(\mathsf{A}, \mathsf{C})}$ occurs, then we know with certainty (in the deterministic sense) that  intrinsic initial-value privacy of node $i$ is preserved generically over the network structure $(\G ,\G_{\textnormal S})$.
  \item[(ii)] If intrinsic initial-value privacy of node $i$ is preserved generically, then $\mathcal{E}_{(\mathsf{A}, \mathsf{C})}$
 must occur with probability one.
\end{itemize}
\end{theorem}
\medskip

Theorem \ref{Theo-verify} indeed indicates a two-step approach to verify whether the intrinsic initial-value privacy of individual nodes is preserved generically. Given any set ${{\mathrm{P}}}\subset \V $ and node $i\in\V$, the first step is to compute the maximal rank $n_{\rm P}^{ob}$ of $\mathbf{O}_{ob}\Eb_{\bar {\mathrm{P}}}$. In practice, since $\rank(\mathbf{O}_{ob}\Eb_{\bar {\mathrm{P}}})=n_{\rm P}^{ob}$ holds for almost all configurations $(\Ab,\Cb)$ complying with the network structure $(\G ,\G_{\textnormal S})$, the value $n_{\rm P}^{ob}$ can be obtained by computing the maximal $\rank(\mathbf{O}_{ob}\Eb_{\bar {\mathrm{P}}})$ under a few independently and randomly generated configurations $(\Ab,\Cb)$. The second step is to verify whether the event $\mathcal{E}_{(\mathsf{A}, \mathsf{C})}$ occurs under the independently and randomly generated configurations. If the event $\mathcal{E}_{(\mathsf{A}, \mathsf{C})}$ occurs, one then conclude that the intrinsic initial-value privacy of node $i$ is preserved generically. Otherwise, if  the event $\mathcal{E}_{(\mathsf{A}, \mathsf{C})}$ does not occur,  the  intrinsic initial-value privacy of node $i$ is lost generically by Theorem \ref{Theorem-GIVP}.

\medskip
\begin{remark}
When $\mathrm{P}=\emptyset$, according to Theorem \ref{Theorem-LIPP},  the generic intrinsic initial-value privacy of node $i$  indicates that the node state $x_{i,t}$ is unobservable for almost all configurations $(\Ab,\Cb)$ complying with $(\G ,\G_{\textnormal S})$, i.e., structurally unobservable, that is an extension of the dual notion to structural state variable controllability \cite{Blackhall2010}, and  structural observability \cite{Willems-1986,Chang-1992}.
\end{remark}

\medskip
Similar to Theorem \ref{Theorem-GIVP}, the network privacy index is also generically determined by the network structure $(\G ,\G_{\textnormal S})$.
\medskip
\begin{theorem}\label{Theorem-GNPI}
    The network privacy index is generically determined by the network topology. Namely, for almost all configurations $(\Ab,\Cb)$ complying with the network structure $(\G ,\G_{\textnormal S})$, the network privacy index ${\mathbf I_{rp}} = n - n_{ob}^{g} -1$ with $n_{ob}^{g}$ given by the maximal rank of the observability matrix $\mathbf{O}_{ob}$.
\end{theorem}

\medskip
\begin{remark}
By the proof of Theorem \ref{Theorem-GNPI}, $\rank(\mathbf{O}_{ob})=n_{ob}^{g}$  holds for almost all configurations $(\Ab,\Cb)$ complying with $(\G ,\G_{\textnormal S})$.
According to \cite{Hosoe-1981} and the duality principle between controllability and observability, the $n_{ob}^{g}$ indeed is given by the maximal number of edges in the set of stem-cycle disjoint graphs \cite{Blackhall2010,Hosoe-1981}.
\end{remark}

\medskip

\begin{remark}
We remark that the results in Theorems \ref{Theorem-LIPP}--\ref{Theorem-GNPI} and Proposition \ref{Proposition-NPI}  can be easily extended to time-varying networked systems (\ref{eq:ini-ltvsys}) under time-varying graphs $(\G_t,\G_{\textnormal S,t})$, by replacing the observability matrix $\mathbf{O}_{ob}$ by its time-varying version $\widehat{\mathbf{O}}$.
\end{remark}

\medskip
\noindent{\bf Example 5}. Now an example is presented to illustrate   Theorems \ref{Theo-verify}-\ref{Theorem-GNPI}. We consider a networked system (\ref{eq:ini-sys}) with $(\Ab,\Cb)$ complying with the graphs $\G,\G_{\textnormal S}$ in Fig. \ref{fig:network-g}, consisting of 12 network nodes and 2 sensing nodes.
\begin{figure}[H]
	\hspace*{0cm}
	\vspace*{0cm}
	\centering
	\includegraphics[width=8.5cm]{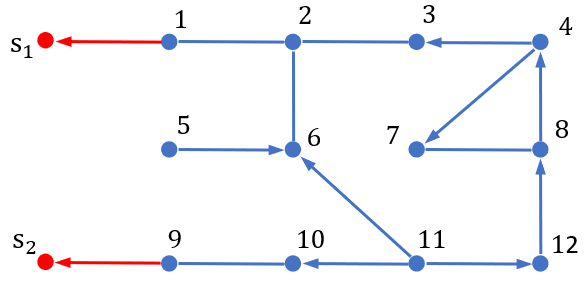}
	\caption{Network topologies $(\G,\G_{\textnormal S})$ with 12 network nodes (blue circles) and 2 sensing nodes (red circles). The lines without arrows denote bidirectional edges.
}
	\label{fig:network-g}
\end{figure}

We choose five configurations $(\mathbf{A},\mathbf{C})$ with the entries generated independently and randomly according to the uniform distribution over the interval $[0,1]$, and complying with the graph in Fig. \ref{fig:network-g}. Under these configurations, we check the corresponding ranks of matrix $\mathbf{O}_{ob}$, and find that the maximal rank is $n_{ob}^{g}=10$.
According to Theorem \ref{Theorem-GNPI}, the network privacy index is ${\mathbf I_{rp}}=1$, for almost all configurations complying with the graph in Fig. \ref{fig:network-g}. Then, for each $i\in\mathrm{V}$ by fixing all edge weights as 1,
it can be verified that
\[
\mbox{rank}\left(\begin{bmatrix} \mathbf{O}_{ob} \cr \eb_i^\top  \end{bmatrix}\right) = 11\,,\quad \mbox{for $i=3,6,7,12$}.
\]
Using Theorem \ref{Theo-verify}, this indicates that the network system preserves the intrinsic initial-value privacy of nodes $3,6,7,12$ for almost all configurations $(\Ab,\Cb)$ complying with the graph in Fig. \ref{fig:network-g}.

Let the public disclosure set ${\mathrm{P}}=\{3,6\}$.
Under the previously generated five configurations $(\mathbf{A},\mathbf{C})$, we check the corresponding ranks of matrix $\mathbf{O}_{ob} \Eb_{\bar {\mathrm{P}}}$, and find that the maximal rank is $n_{\rm P}^{ob}=9$.
Then fix all edge weights as 1 in Fig. \ref{fig:network-g}, and we can obtain
\[
\mbox{rank}\left(\begin{bmatrix} \mathbf{O}_{ob} \Eb_{\bar {\mathrm{P}}} \cr \eb_i^\top\Eb_{\bar {\mathrm{P}}} \end{bmatrix}\right) = 10\,,\quad \mbox{for $i=7,12$}.
\]
By Theorem \ref{Theo-verify}, this indicates that the network system preserves the intrinsic initial-value privacy of nodes $7,12$ for almost all configurations $(\Ab,\Cb)$ complying with the graph in Fig. \ref{fig:network-g}, under the the public disclosure set ${\mathrm{P}}=\{3,6\}$.

Furthermore, let ${\mathrm{P}}=\{3,7\}$. Similarly, we check the corresponding ranks of matrix $\mathbf{O}_{ob} \Eb_{\bar {\mathrm{P}}}$ under the previously generated five configurations $(\mathbf{A},\mathbf{C})$, and can find that  $n_{\rm P}^{ob}=10$, i.e., $\mathbf{O}_{ob} \Eb_{\bar {\mathrm{P}}} $ is full-column-rank for almost all configurations $(\Ab,\Cb)$ complying with the graph in Fig. \ref{fig:network-g}. This implies that for all $i$, there is no configuration $(\Ab,\Cb)$ such that
\[
\mbox{rank}\left(\begin{bmatrix} \mathbf{O}_{ob} \Eb_{\bar {\mathrm{P}}} \cr \eb_i^\top\Eb_{\bar {\mathrm{P}}} \end{bmatrix}\right) =  11 \,.
\]
Thus, by Theorem \ref{Theo-verify} the intrinsic initial-value privacy of all nodes is generically lost. $\blacksquare$

\medskip

\begin{remark}
  We remark that for a large-scale system (\ref{eq:ini-sys}), it is generally difficult to determine the generic intrinsic initial-value privacy of individual nodes, due to the high computational complexity to compute the maximal rank of $\mathbf{O}_{ob}\Eb_{\bar{\mathrm{P}}}$ and verify the rank conditions in (C1)-(C3). A possible solution is to find topological conditions as in \cite{Lin-1974,Julien-19}. However, the extension of ideas in \cite{Lin-1974,Julien-19} to our settings is nontrivial, particularly in the presence of a public disclosure set $\mathrm{P}$. This will be further explored in future works.

\end{remark}


\section{Conclusions}
\label{sec-6}
In this paper, we have studied the intrinsic initial-value privacy and differential initial-value privacy  of linear dynamical systems with random process and measurement noises. We proved that the intrinsic initial-value privacy is equivalent to unobservability, while the  differential initial-value privacy can be  achieved for a privacy budget  depending on an extended observability matrix of the system and the covariance of the noises.
Next, by regarding the considered linear system as a network system, we proposed  necessary and sufficient conditions on the intrinsic initial-value privacy of individual nodes, in the presence of some nodes whose initial-value privacy is public. A quantitative network privacy index was also proposed using the largest number of arbitrary public nodes such that the whole initial values are not fully exposed. In addition, we showed that both the intrinsic initial-value privacy and the network privacy index are generically determined by the network structure.  In future works, topological conditions (see e.g. \cite{Julien-19,Lin-1974}) will be explored for generic intrinsic initial-value privacy of individual nodes, and the considered privacy metrics will be utilized to develop privacy-preservation approaches for linear dynamical systems.

\appendix

\section{Proof of Proposition \ref{Proposition-IP}}

\noindent
{\emph{Sufficiency.}} If $\rank(\mathbf{O}_{ob})<n$, then $\rank(\mathbf{O}_{T})<n$ and for all $\etab\in\ker(\mathbf{O}_{T})$, we have
\begin{equation}\label{eq:IP-eta}
  \pdf\left(\mathbf{Y}_T| \xb_0\right) = \pdf\left(\mathbf{Y}_T | \xb_0 + \etab\right)\,.
\end{equation}
This, according to Definition \ref{Definition:IP}, completes the sufficiency.

\medskip
\noindent
{\emph{Necessity.}} To show the necessity part, the contradiction method is used. With the system (\ref{eq:ini-sys}) preserving intrinsic initial-value privacy, we suppose that  $\rank(\mathbf{O}_T)=n$. This implies that the mapping $\mathcal{M}(\cdot)$ is invertible, i.e.,
\[
\xb_0 =  \mathcal{M}^{-1}(\mathbf{Y}_T) := \big(\mathbf{O}_T^\top\mathbf{O}_T\big)^{-1}\mathbf{O}_T^\top\big(\mathbf{Y}_T - \mathbf{H}_T \mathbf{V}_T - \mathbf{W}_T\big)\,.
\]
It is clear that $\xb_0$ is identifiable from $\mathbf{Y}_T$, i.e., the initial value $\xb_0$ of the system (\ref{eq:ini-sys}) is not private.
This contradicts with the fact that the system (\ref{eq:ini-sys}) preserves intrinsic initial-value privacy. Therefore, it can be concluded that $\rank(\mathbf{O}_T)<n$. This completes the proof.

\section{Proof of Theorem \ref{Theorem-DP}}

Let $\mathbf{V}_T^i, \mathbf{W}_T^i$ be the process and measurement noise vectors, respectively for the $i$-th implementation. Define $\zb^i:=\mathbf{H}_T \mathbf{V}_T^i + \mathbf{W}_T^i$ and $\zb:=[\zb^1;\zb^2;\ldots;\zb^N]$. Hence, $\zb\backsim \mathcal{N}(0,\Ib_N\otimes\Sigma)$ with
\[
\Sigma:=\begin{bmatrix}
                                       \mathbf{H}_T & \Ib_{m(T+1)}
                                     \end{bmatrix} \Sigma_T
                                     \begin{bmatrix}
                                       \mathbf{H}_T & \Ib_{m(T+1)}
                                     \end{bmatrix}^\top\,.
                                     \]
As in  \cite{Ny-Pappas-TAC-2014}, with the mapping $\mathcal{M}^N$ defined in (\ref{eq:M^N}), for any $R\subset\mbox{range}(\mathcal{M}^N)$ and $\epsilon>0$, we have
\[\ba{l}
\vspace{2mm}
\Pbb(\mathcal{M}^N(\xb_0)\in R) = \Pbb\big((\mathbf{1}_N\otimes\mathbf{O}_T)\xb_0 + \zb\in R\big)\,\\
\vspace{2mm}
\overset{a)}{=}(2\pi)^{-\frac{Nm(T+1)}{2}}\det(\Sigma)^{-\frac{N}{2}}\dst\int_{R}\exp\bigg(-\frac{1}{2}\|(\Ib_N\otimes\Sigma^{-\frac{1}{2}})(u-(\mathbf{1}_N\otimes\mathbf{O}_T)\xb_0)\|^2\bigg)\,\d u\,\\
\vspace{2mm}
\overset{b)}{=} \exp(\epsilon)\Pbb(\mathcal{M}^N(\xb_0^\prime)\in R) + (2\pi)^{-\frac{Nm\,(T+1)}{2}}\det(\Sigma)^{-\frac{N}{2}}\dst\int_{R}\exp\left(-\frac{1}{2}\|(\Ib_N\otimes\Sigma^{-\frac{1}{2}})(u-(\mathbf{1}_N\otimes\mathbf{O}_T)\xb_0)\|^2\right)\cdot \\
\vspace{2mm}
\qquad \cdot \left(1-\exp\bigg(\epsilon - (u-(\mathbf{1}_N\otimes\mathbf{O}_T)\xb_0)^\top(\mathbf{1}_N\otimes\Sigma^{-1}\mathbf{O}_T)\tilde\xb_0-\frac{N}{2}\tilde\xb_0^\top\mathbf{O}_T^\top\Sigma^{-1}\mathbf{O}_T\tilde\xb_0\bigg)\right)\,\d u\,\\
\vspace{2mm}
\overset{c)}{\leq} \exp(\epsilon)\Pbb(\mathcal{M}^N(\xb_0^\prime)\in R) + (2\pi)^{-\frac{Nm\,(T+1)}{2}}\dst\int_{\R^{Nm(T+1)}}1_{ v^T(\mathbf{1}_N\otimes\Sigma^{-\frac{1}{2}}\mathbf{O}_T)\tilde\xb_0 \geq \epsilon- \frac{N}{2}\|\Sigma^{-\frac{1}{2}}\mathbf{O}_T\tilde\xb_0\|^2} \exp\big(-\frac{\|v\|^2}{2}\big)\,\d v\,\\
\vspace{2mm}
\overset{d)}{\leq} \exp(\epsilon)\Pbb(\mathcal{M}^N(\xb_0^\prime)\in R)  + \Pbb\bigg( \tilde\xb_0^\top \left(\mathbf{1}_N^\top\otimes\mathbf{O}_T^\top \Sigma^{-\frac{1}{2}}\right)v \geq \epsilon- \frac{N}{2}\|\Sigma^{-\frac{1}{2}}\mathbf{O}_T\tilde\xb_0\|^2\bigg)\,
\ea\]
where $a)$ is obtained by setting $u=(\mathbf{1}_N^\top\otimes\mathbf{O}_T)\xb_0 + \zb$, $b)$ is obtained by adding and subtracting the term $\exp(\epsilon)\Pbb(\mathcal{M}^N(\xb_0^\prime)\in R)$, and setting $\tilde\xb_0=\xb_0-\xb^\prime_0$,  $c)$ is obtained by choosing $R=\R^{Nm(T+1)}$ and setting $v=(\Ib_{N}\otimes\Sigma^{-\frac{1}{2}})[u-(\mathbf{1}_N\otimes\mathbf{O}_T)\xb_0]$ with $v\backsim \mathcal{N}(0,\Ib_{Nm(T+1)})$, and d) is derived by using the fact that $\tilde\xb_0^\top\left(\mathbf{1}_N^\top\otimes\mathbf{O}_T^\top \Sigma^{-\frac{1}{2}}\right)v\backsim \mathcal{N}\big(0,N\|\Sigma^{-\frac{1}{2}}\mathbf{O}_T\tilde\xb_0\|^2\big)$. It is noted that
\[\ba{rcl}
\vspace{2mm}
&&\Pbb\bigg( \tilde\xb_0^\top \left(\mathbf{1}_N^\top\otimes\mathbf{O}_T^\top \Sigma^{-\frac{1}{2}}\right)v \geq \epsilon- \frac{N}{2}\|\Sigma^{-\frac{1}{2}}\mathbf{O}_T\tilde\xb_0\|^2\bigg) \\
&\leq& \Pbb\bigg( Z \geq \frac{\epsilon\sqrt{\sigma_m(\Sigma)}}{d\sqrt{N}\|\mathbf{O}_T\|}- \frac{d\sqrt{N}\|\mathbf{O}_T\|}{2\sqrt{\sigma_m(\Sigma)}}\bigg)\,\\
\ea
\]
with $Z\backsim \mathcal{N}(0,1)$.

Thus, it can be concluded that the $(\epsilon,\delta)$-differential privacy is preserved if $\epsilon,\delta$ satisfy
\beeq{\label{eq:delta}
\delta \geq \mathcal{Q}\bigg(\frac{\epsilon\sqrt{\sigma_m(\Sigma)}}{d\sqrt{N}\|\mathbf{O}_T\|}- \frac{d\sqrt{N}\|\mathbf{O}_T\|}{2\sqrt{\sigma_m(\Sigma)}}\bigg)\,\,.
}
With $\mathcal{Q}(w)$ being a strictly decreasing smooth function, it is clear that (\ref{eq:delta}) is equivalent to
\[
\frac{\epsilon\sigma_m(\Sigma)}{d\sqrt{N}\|\mathbf{O}_T\|}- \mathcal{Q}^{-1}(\delta)\sqrt{\sigma_m(\Sigma)}- \frac{d\sqrt{N}\|\mathbf{O}_T\|}{2} \geq 0\,,
\]
which is fulfilled if (\ref{eq:sigma}) holds.
The proof is thus completed.

\section{Proof of Theorem \ref{Theorem-LIPP}}
\label{app-sec-theo-2}

Let $\Kb_{j}$ be the $j$-the column of matrix $\mathbf{O}_T$. Fundamental to the proof is the following technical lemma.
\begin{lemma}
Statement b) and c) in Theorem \ref{Theorem-LIPP} are respectively equivalent to the following b$^\dag$) and c$^\dag$).
  \begin{itemize}
    \item[b$^\dag$).] $\rank([\Kb_{\bar p_1},\ldots,\Kb_{\bar p_{n-l}}]) = \rank([\Kb_{i_1},\ldots,\Kb_{i_{n-l-1}}])
  $ with $\{i_1,\ldots,i_{n-l-1}\}=\V \backslash({\mathrm{P}}\cup \{i\})$.
  \vspace{2mm}
  \item[c$^{\dag}$).] $\rank\left(\begin{bmatrix}  \mathbf{O}_{T}\cr \eb_i^\top \end{bmatrix}\Eb_{\bar {{\mathrm{P}}} }\right) = \rank\left(\mathbf{O}_{T}\Eb_{\bar {{\mathrm{P}}} }\right) + 1$.
  \end{itemize}
\end{lemma}
\begin{proof}
  We observe that $\mathbf{O}_{T}=\mathbf{Q}_{T}\mathbf{O}_{ob}$ with $\mathbf{Q}_{T}\in\mathbb{R}^{m(T+1)\times mn}$ being full-column-rank, which yields
  \beeq{\label{eq:rank_OE}
  \rank(\mathbf{O}_{T}\Eb_{\bar{\mathrm{P}}}) = \rank(\mathbf{O}_{ob}\Eb_{\bar{\mathrm{P}}})\,.
  }
  Namely,
  \[
  \rank([\Kb_{\bar p_1},\ldots,\Kb_{\bar p_{n-l}}]) = \rank([\Kb_{\bar p_1}^{\textnormal ob},\ldots,\Kb_{\bar p_{n-l}}^{\textnormal ob}])\,.
   \]
  Similarly, it can be verified that $\rank([\Kb_{i_1},\ldots,\Kb_{i_{n-l-1}}]) = \rank([\Kb_{i_1}^{\textnormal ob},\ldots,\Kb_{i_{n-l-1}}^{\textnormal ob}])$. Therefore, the statements b) and  b$^\dag$) are equivalent.

  To show the equivalence between statements c) and  c$^\dag$), we observe that
  \[\ba{l}
  \vspace{2mm}
  \rank\left(\begin{bmatrix}  \mathbf{O}_{ob}\cr \Eb_{\mathrm{P}}^\top\end{bmatrix} \right) = \rank\left(\begin{bmatrix}  \mathbf{O}_{T}\cr \Eb_{\mathrm{P}}^\top \end{bmatrix} \right) = \rank\left(\mathbf{O}_{T}\Eb_{\bar {{\mathrm{P}}} }\right)+l\,\,\\
  \vspace{2mm}
  \rank\left(\begin{bmatrix}  \mathbf{O}_{ob}\cr \Eb_{\mathrm{P}}^\top \cr \eb_i^\top\end{bmatrix} \right) = \rank\left(\begin{bmatrix}  \mathbf{O}_{T}\cr \Eb_{\mathrm{P}}^\top \cr \eb_i^\top\end{bmatrix} \right)\\ \qquad \qquad \qquad \,\quad = \rank\left(\begin{bmatrix}  \mathbf{O}_{T}\cr \eb_i^\top \end{bmatrix}\Eb_{\bar {{\mathrm{P}}} }\right)+l\,.
  \ea\]
  This thus establishes the equivalence between statements  $c)$ and c$^\dag$). $\blacksquare$
\end{proof}

With the above lemma in mind, we now proceed to show $a)\Longleftrightarrow b^\dag)$ and $b)\Longleftrightarrow c^\dag)$, which in turn will complete the proof.

\medskip
\noindent
{$a)\Longrightarrow b^\dag)$.} Given any node $i\in\V$, we now proceed to show, if the initial-value privacy of node $i$ is preserved w.r.t. ${{\mathrm{P}}}$,  $\Kb_{i}\in \range([\Kb_{i_1},\Kb_{i_{2}},\ldots,\Kb_{i_{n-l-1}}])$. The contradiction method is used by supposing there exists no vector $\Gamma\in\R^{n-l-1}$ such that $[\Kb_{i_1},\Kb_{i_2},\ldots,\Kb_{i_{n-l-1}}]\Gamma=\Kb_{i}$. This, in turn, implies there exists a vector $\Db_i\in\R^{m(T+1)}$ such that $\Db_i^\top\Kb_{i}\neq 0$ and $\Db_i^\top\Kb_{i_j} = 0$ for all $j=1,\ldots,n-l-1$. Thus, given any initial condition $\xb_0=[x_{1,0};\ldots;x_{n,0}]$, we have
\[
\Db_i^\top\mathbf{Y}_T = \sum\limits_{j\in{{\mathrm{P}}}}\Db_i^\top\Kb_j x_{j,0} + \Db_i^\top\Kb_i x_{i,0} + \Db_i^\top(\mathbf{H}_T \mathbf{V}_T + \mathbf{W}_T)\,,
\]
which implies
\[\ba{rcl}
\vspace{2mm}
x_{i,0} &=& (\Db_i^\top\Kb_i)^{-1}\Db_i^\top\mathbf{Y}_T-\sum\limits_{j\in{{\mathrm{P}}}}\Db_i^\top\Kb_j x_{j,0} \\ &&-(\Db_i^\top\Kb_i)^{-1}\Db_i^\top\mathbf{H}_T \mathbf{V}_T - (\Db_i^\top\Kb_i)^{-1}\Db_i^\top\mathbf{W}_T\,.
\ea\]
It is clear that the initial condition $x_{i,0}$ of node $i$ is identifiable for system (\ref{eq:ini-sys}) w.r.t. ${{\mathrm{P}}}$, which contradicts with the fact that $x_{i,0}$ is private. Therefore, we conclude that $\Kb_{i}\in \range([\Kb_{i_1},\Kb_{i_{2}},\ldots,\Kb_{i_{n-l-1}}])$, i.e., $a)$ implies $b^\prime)$.

\medskip
\noindent
{$b^\dag)\Longrightarrow a)$.}
Clearly, if
\[
\rank([\Kb_{i},\Kb_{i_1},\ldots,\Kb_{i_{n-l-1}}]) = \rank([\Kb_{i_1},\ldots,\Kb_{i_{n-l-1}}])\,,
 \]
then there exists a vector $\Gamma= [\gamma_1;\ldots;\gamma_{n-l-1}]$ such that $\Kb_{i}=\sum\limits_{j=1}^{n-l-1}\gamma_j\Kb_{i_j}$. Given any $\xb_0=[x_{1,0};\ldots;x_{n,0}]$ and $\xb^\prime_0=[x_{1,0}^\prime;\ldots;x_{n,0}^\prime]$ with $x_{i,0}\neq x_i^\prime(0)$, $x_{j,0}=x_{j,0}^\prime$ for $j\in{{\mathrm{P}}}$, and $x_{i_j,0}^\prime$, $j=1,\ldots,n-l-1$ satisfying
\[
x_{i_j,0}^\prime = x_{i_j,0} + \gamma_j(x_{i,0}-x_{i,0}^\prime)\,,
\]
we can obtain
\[\ba{l}
\mathbf{Y}_T = \sum\limits_{j=1}^n\Kb_j x_{j,0} + \mathbf{H}_T \mathbf{V}_T + \mathbf{W}_T\,\\
= \sum\limits_{j\in\mathrm{P}}\Kb_j x_{j,0}  + \Kb_i x_{i,0} + \sum\limits_{j=1}^{{n-l-1}}\Kb_{i_j} x_{i_j,0} + \mathbf{H}_T \mathbf{V}_T + \mathbf{W}_T\\
= \sum\limits_{j\in\mathrm{P}}\Kb_j x_{j,0}^\prime  + \Kb_i x_{i,0}^\prime + \sum\limits_{j=1}^{n-l-1}\gamma_j\Kb_{i_j}(x_{i,0}-x_{i,0}^\prime)\,   +  \sum\limits_{j=1}^{{n-l-1}}\Kb_{i_j} x_{i_j,0}^\prime - \sum\limits_{j=1}^{{n-l-1}}\Kb_{i_j}\gamma_j(x_{i,0}-x_{i,0}^\prime)  + \mathbf{H}_T \mathbf{V}_T + \mathbf{W}_T\,\\
=\sum\limits_{j=1}^n\Kb_j x_{j,0}^\prime + \mathbf{H}_T \mathbf{V}_T + \mathbf{W}_T\,.
\ea\]
Thus, given any nonzero $\eta_i$, we let $\eta_j=0$ for all $j\in{{\mathrm{P}}}$ and $\eta_{i_j}=\gamma_j\eta_i$ for all $j=1,\ldots,n-l-1$, we have
\begin{equation}\label{eq:IPP-eta}
  \pdf\left(\mathbf{Y}_T| \xb_0\right) = \pdf\left(\mathbf{Y}_T | \xb_0 + \etab\right)\,
\end{equation}
with $\etab:=\col(\eta_1,\ldots,\eta_n)$.
This, according to Definition \ref{Definition:LIPP}, proves $b^\dag)\Longrightarrow a)$.

\medskip
\noindent
{$b^\dag)\Longleftrightarrow c^\dag)$.} The equivalence between $b^\dag)$ and $c^\dag)$ can be easily inferred by using the facts that
\[\ba{l}
\vspace{2mm}
\rank\left(\begin{bmatrix}  \mathbf{O}_{T} \cr \eb_i^\top \end{bmatrix}\Eb_{\bar {{\mathrm{P}}} }\right) = \rank\left(\begin{bmatrix} \Kb_i &\Kb_{i_1}&\ldots&\Kb_{i_{n-l-1}} \cr 1 & 0 & \cdots & 0 \end{bmatrix}\right)\,\\
=\rank\left([\Kb_{i_1},\ldots,\Kb_{i_{n-l-1}}]\right) + 1\,.
\ea\]

\section{Proof of Proposition \ref{Proposition-NPI}}

The proof is completed if the following two statements are proved.
\begin{itemize}
  \item[(i)]  Given any ${{\mathrm{P}}}\subset\mathrm{V} $ with $|{{\mathrm{P}}}|=n-\rank(\mathbf{O}_{ob})-1$, there exists a node whose initial value is private.
  \item[(ii)] There exists a set ${{\mathrm{P}}}\subset\mathrm{V} $ with $|{{\mathrm{P}}}|=n-\rank(\mathbf{O}_{ob})$ such that initial values of all  nodes are identifiable.
\end{itemize}

\medskip
\noindent
\emph{Proof of part (i).} Given any ${{\mathrm{P}}}\subset\mathrm{V} $ with $|{{\mathrm{P}}}|:=l=n-\rank(\mathbf{O}_{ob})-1$, it is clear that the matrix
$
[\Kb_{\bar p_1}^{ob},\Kb_{{\bar p_2}}^{ob},\ldots,\Kb_{\bar p_{n-l}}^{ob}]
$
is not full-column-rank, which indicates
there exits an $i=p_j\in\{\bar p_1,\bar p_2,\ldots,\bar p_{n-l}\}$ such that
\[
\Kb_{i}^{ob}\in \range([\Kb_{\bar p_1}^{ob},\ldots,\Kb_{\bar p_{j-1}}^{ob},\Kb_{\bar p_{j+1}}^{ob},\ldots,\Kb_{\bar p_{n-l}}^{ob}])\,.
\]
According to Theorem \ref{Theorem-LIPP}, this yields that the initial value of node $i$ is private.

\medskip
\noindent
\emph{Proof of part (ii).} Let $r=\rank(\mathbf{O}_{ob})$ and $l=n-r$. Then, select $\bar p_j\in\V $, $j=1,\ldots,r$  such that  $\rank([\Kb_{\bar p_1}^{ob},\Kb_{\bar p_2}^{ob},\ldots,\Kb_{\bar p_r}^{ob}])=r$. Thus let ${{\mathrm{P}}}=\mathrm{V} \backslash\{\bar p_1,\bar p_2,\ldots,\bar p_r\}$. According to Theorem \ref{Theorem-LIPP}, it can be concluded that the initial values of all nodes $i\in\{\bar p_1,\bar p_2,\ldots,\bar p_r\}=\mathrm{V}\backslash{{\mathrm{P}}}$ are identifiable.

\medskip
In summary of the previous analysis, we conclude that ${\mathbf I_{rp}} = n-\rank(\mathbf{O}_{ob})-1$ for system (\ref{eq:ini-sys}).

\section{Proof of Theorem 3}
\label{app-theo-3}

To ease the subsequent analysis, we collect all edge weights $a_{ij},c_{ij}$ in a configuration vector $\theta\in\R^N$ with $N$ being the total number of edges in $(\G,\G_{\textnormal S})$. In this way, all matrices $\Ab,\Cb$ are indeed functions of $\theta$, and so is the resulting observability matrix $\mathbf{O}_{ob}$.


Instrumental to the proof is the following lemma.
\medskip
\begin{lemma}\label{lemma-1}
There exists a  $n_{\rm P}^{ob}\leq n-l$ such that
   \begin{equation}\label{eq:E1}
   \rank(\mathbf{O}_{ob}(\theta)\Eb_{\bar{\mathrm P}}) = n_{\rm P}^{ob}\,, \quad \mbox{for almost all $\theta\in\R^N$}
   \end{equation}
   \begin{equation}\label{eq:E2}
     \rank(\mathbf{O}_{ob}(\theta)\Eb_{\bar{\mathrm P}}) \leq n_{\rm P}^{ob}\,,\quad \mbox{for all $\theta\in\R^N$}.\qquad\,\,\,\,
   \end{equation}
\end{lemma}
\begin{proof}
  The proof of this lemma is to  find the maximal value of  $\mbox{rank}(\Phi(\theta))$ with $\Phi(\theta):=\mathbf{O}_{ob}(\theta)\Eb_{\bar{\mathrm P}}$. Let $\bar\alpha_j(\theta):\R^N\rightarrow\R$, $j=1,\ldots,n-l$ be such that
\[
\det\left(s\Ib - \Phi(\theta)^\top\Phi(\theta)\right)
= \sum_{j=1}^{n-l}\bar\alpha_j(\theta) s^{j-1} + s^{n-l}\,.
\]
We run the following recursive algorithm from $k=1$ until $n_{\rm P}^{ob}$ is found.

\medskip
{\bf Step} $k${\bf :} Check whether there exists $\theta^\prime\in\R^N$ such that $\bar\alpha_k(\theta^\prime)\neq 0$. If so, using the fact that analytic functions that are not identically zero vanish only on a zero-measure set, we can conclude that  $\bar\alpha_k(\theta)\neq 0$ holds for almost all $\theta\in\R^N$. This, together with the fact that $\bar\alpha_j(\theta)=0$ for all $j\leq k-1$ and all $\theta\in\R^N$, indicates that (\ref{eq:E1}) and (\ref{eq:E2}) hold with $n_{\rm P}^{ob}=n-l-k+1$. Otherwise, if for all $\theta\in\R^N$, $\bar\alpha_k(\theta)= 0$, we then  proceed to Step $k+1$.

\medskip
If at the $n-l$-th recursion of the above algorithm, we still cannot find a $\theta\in\R^N$ such that $\bar\alpha_{n-l}(\theta)\neq 0$, we then can conclude that $n_{\rm P}^{ob} =0$. $\blacksquare$
\end{proof}

With this lemma, we now proceed to prove the theorem. Let $\Theta_{1}\subseteq\R^N$ be a set of configuration vector $\theta$ such that the intrinsic initial-value privacy of node $i$ is preserved for all $\theta\in\Theta_{1}$ and lost for all $\theta\in\R^N\backslash \Theta_{1}$.
It is clear that the proof is done if we show that either $\Theta_{1}$ or $\R^N\backslash \Theta_{1}$ is zero-measure. To prove it, we use the contradiction method, and assume that there exists a nonzero-measure set $\Theta_1\in\R^N$ of configuration vector $\theta$ such that
\begin{itemize}
    \item[($P1$)] the set $\R^N\backslash\Theta_{1}$ is nonzero-measure, and
    \item[($P2$)] the intrinsic initial-value privacy of node $i$ is preserved only for  $\theta\in\Theta_1$ under $(\G ,\G_{\textnormal S})$, and
    \item[($P3$)] the intrinsic initial-value privacy of node $i$ is lost  for  $\theta\in\R^N\backslash\Theta_{1}$ under $(\G,\G_{\textnormal S})$.
  \end{itemize}
Then, according to Theorem \ref{Theorem-LIPP} and Remark \ref{Remark-9}, it can be inferred that
\[\ba{l}
\vspace{1mm}
(P2)
\Longleftrightarrow 
\rank\begin{bmatrix}
                                 \mathbf{O}_{ob}(\theta)\Eb_{\bar{\mathrm P}}\\
                                 \eb_i^\top \Eb_{\bar{\mathrm P}}
                               \end{bmatrix}
                              = \rank\left(\mathbf{O}_{ob}(\theta)\Eb_{\bar{\mathrm P}}\right) +1 \,, \mbox{for all } \theta\in\Theta_1\,.
\ea\]
Let $\alpha_j(\theta):\R^N\rightarrow\R$, $j=1,\ldots,n$ be such that
\[\ba{l}
\det\left(s\Ib - \begin{bmatrix}
                                  \mathbf{O}_{ob}(\theta)\Eb_{\bar{\mathrm P}}\\
                                 \eb_i^\top\Eb_{\bar{\mathrm P}}
                               \end{bmatrix}^\top\begin{bmatrix}
                                 \mathbf{O}_{ob}(\theta)\Eb_{\bar{\mathrm P}} \\\eb_i^\top\Eb_{\bar{\mathrm P}}
                               \end{bmatrix}\right)
 = \sum\limits_{j=1}^{n-l}\alpha_j(\theta) s^{j-1} + s^{n-l}\,.
\ea\]

By  Lemma \ref{lemma-1},  there exists a nonzero-measure set $\Theta_{ob}\subseteq\R^N$  such that  the set $\R^N\backslash\Theta_{ob}$ is zero-measure, and for all $\theta\in\Theta_{ob}$, $\rank(\mathbf{O}_{ob}(\theta)\Eb_{\bar{\mathrm P}}) = n_{\rm P}^{ob}$. Besides, it is clear that, for all $\theta\in\R^N$
\beeq{\label{eq:app-E-lower}
\rank\left(\begin{bmatrix}
\mathbf{O}_{ob}(\theta)\Eb_{\bar{\mathrm P}}\\
\eb_i^\top\Eb_{\bar{\mathrm P}}
\end{bmatrix}\right)
                               \leq  \rank\left(\mathbf{O}_T(\theta)\Eb_{\bar{\mathrm P}}\right)+1\,.
}

Since $\Theta_{1}$ is nonzero-measure and $\R^N\backslash\Theta_{ob}$ is zero-measure, we have $\Theta_{1}\cap\Theta_{ob}\neq \emptyset$. Thus letting $\theta^\ast\in\Theta_{1}\cap\Theta_{ob}$ yields that  $\rank(\mathbf{O}_{ob}(\theta^\ast)\Eb_{\bar{\mathrm P}}) = n_{\rm P}^{ob}$ and
\[\ba{rcl}
\rank\begin{bmatrix}
                                 \mathbf{O}_{ob}(\theta^\ast)\Eb_{\bar{\mathrm P}}\\
                                 \eb_i^\top \Eb_{\bar{\mathrm P}}
                               \end{bmatrix}
                              &=& \rank(\mathbf{O}_{ob}(\theta^\ast)\Eb_{\bar{\mathrm P}})+ 1\,\\
                              &=& n_{\rm P}^{ob}+1\,.
\ea\]
This then implies $\alpha_{n_{uo}}(\theta^\ast)\neq 0$ with $n_{uo}=n-l-n_{\rm P}^{ob}$.
Note that analytic functions that are not identically zero vanish only on a zero-measure set. This indicates that there is a nonzero-measure set $\Theta_2\subseteq\R^N$ of configuration vector $\theta$ such that
\begin{itemize}
    \item[($P4$)] the set $\R^N\backslash\Theta_{2}$ is zero-measure, and
    \item[($P5$)] the inequality $\alpha_{n_{uo}}(\theta)\neq 0$ holds for all $\theta\in\Theta_2$.
  \end{itemize}
Thus, for all $\theta\in\Theta_{ob}\cap\Theta_{2}$, we have
\[
\rank\left(\begin{bmatrix}
                                 \mathbf{O}_{ob}(\theta)\Eb_{\bar{\mathrm P}} \\
                                 \eb_i^\top \Eb_{\bar{\mathrm P}}
                               \end{bmatrix}\right)
                               \geq n_{\rm P}^{ob}+1 \,= \rank\left(\mathbf{O}_T(\theta)\Eb_{\bar{\mathrm P}}\right)+1\,,
\]
which, together with (\ref{eq:app-E-lower}), yields
\[
\rank\left(\begin{bmatrix}
                                 \mathbf{O}_{ob}(\theta)\Eb_{\bar{\mathrm P}}\\
                                 \eb_i^\top \Eb_{\bar{\mathrm P}}
                               \end{bmatrix}\right)
                               =  \rank\left(\mathbf{O}_T(\theta)\Eb_{\bar{\mathrm P}}\right)+1\,
\]
for all $\theta\in \Theta_{ob}\cap\Theta_{2}$. According to Theorem \ref{Theorem-LIPP} and Remark \ref{Remark-9}, this implies that the intrinsic initial-value privacy of node $i$ is preserved for all configuration vector $\theta\in \Theta_{ob}\cap\Theta_{2}$. Then by ($P2$) and ($P3$), it immediately follows that $(\R^N\backslash\Theta_{1}) \subseteq \R^N\backslash(\Theta_{ob}\cap\Theta_{2})$, where the set $\R^N\backslash(\Theta_{ob}\cap\Theta_{2})=(\R^N\backslash\Theta_{ob})\cup(\R^N\backslash\Theta_{2})$ is zero-measure. This indicates that $\R^N\backslash\Theta_{1}$ is zero-measure, which contradicts with ($P1$), and thus completes the proof.

\section{Proof of Theorem \ref{Theo-verify}}

Recalling Theorem \ref{Theorem-LIPP} and Remark \ref{Remark-9}, we can easily see that $(C1), (C2)$ and $(C3)$ are equivalent. Thus, the proof is done if the following two statements are proved.
\begin{itemize}
  \item[($S1$)] If the condition ($C1$) holds, then the intrinsic initial-value privacy of node $i$ is preserved generically.
  \item[($S2$)] If the intrinsic initial-value privacy of node $i$ is preserved generically, then the condition ($C1$) holds for almost all configurations complying with $\G,\G_{\rm P}$.
\end{itemize}

\medskip
\noindent
\emph{Proof of ($S1$).}
By the condition (C1), there exists a $\theta^\ast\in\R^N$ such that
\[
\rank\left(\begin{bmatrix}  \mathbf{O}_{ob}(\theta^\ast) \cr \eb_i^\top \end{bmatrix}\Eb_{\bar {{\mathrm{P}}} }\right) = n_{\rm P}^{ob} + 1\,.
\]
Following the notations in section \ref{app-theo-3}, we can obtain
\[
\alpha_{n_{uo}}(\theta^\ast)\neq 0
\]
with $n_{uo}=n-l-n_{\rm P}^{ob}$.  Therefore, according to the standard arguments,  there exists a nonzero-measure set $\Theta_3\subseteq\R^N$  such that the set $\R^N\backslash\Theta_{3}$ is zero-measure, and the inequality $\alpha_{n_{uo}}(\theta)\neq 0$ holds for all $\theta\in\Theta_3$. This yields
\[
\rank\left(\begin{bmatrix}  \mathbf{O}_{ob}(\theta) \cr \eb_i^\top \end{bmatrix}\Eb_{\bar {{\mathrm{P}}} }\right) \geq n_{\rm P}^{ob} + 1\, , \quad\mbox{for all $\theta\in\Theta_3$}\,.
\]
Therefore, with (\ref{eq:E2}) we have
\beeq{\label{eq:app}
\rank\left(\begin{bmatrix}  \mathbf{O}_{ob}(\theta) \cr \eb_i^\top \end{bmatrix}\Eb_{\bar {{\mathrm{P}}} }\right) = n_{\rm P}^{ob} + 1\,
}
for all $\theta\in\Theta_3$. Since $\rank(\mathbf{O}_{ob}(\theta)\Eb_{\bar{\mathrm P}}) = n_{\rm P}^{ob}$ for all $\theta\in\Theta_{ob}$ with $\R^N\backslash\Theta_{ob}$ being zero-measure by Lemma \ref{lemma-1}, we thus obtain
\beeq{\label{eq:app-0}
\rank\left(\begin{bmatrix}  \mathbf{O}_{ob}(\theta) \cr \eb_i^\top \end{bmatrix}\Eb_{\bar {{\mathrm{P}}} }\right) = \rank(\mathbf{O}_{ob}(\theta)\Eb_{\bar{\mathrm P}}) + 1\,
}
for all $\theta\in \Theta_3 \cap \Theta_{ob}$. By Theorem \ref{Theorem-LIPP} and Remark \ref{Remark-9}, this implies that the intrinsic initial-value privacy of node $i$ is preserved for all configurations $\theta\in \Theta_3 \cap \Theta_{ob}$. Note that $\R^N\backslash(\Theta_3 \cap \Theta_{ob})=\R^N\backslash\Theta_3 \cup \R^N\backslash \Theta_{ob}$ is zero-measure, which proves ($S1$).

\medskip
\noindent
\emph{Proof of ($S2$).}
Suppose the intrinsic initial-value privacy of node $i$ is preserved generically. By Theorem \ref{Theorem-LIPP} and Remark \ref{Remark-9}, this implies there exists a  nonzero-measure set $\Theta_{4}\subseteq\R^N$  such that  the set $\R^N\backslash\Theta_{4}$ is zero-measure, and for all $\theta\in\Theta_{4}$, (\ref{eq:app-0}) holds.

Recalling the fact that $\rank(\mathbf{O}_{ob}(\theta)\Eb_{\bar{\mathrm P}}) = n_{\rm P}^{ob}$ for all $\theta\in\Theta_{ob}$, we have
\beeq{\label{eq:app}
\rank\left(\begin{bmatrix}  \mathbf{O}_{ob}(\theta) \cr \eb_i^\top \end{bmatrix}\Eb_{\bar {{\mathrm{P}}} }\right) = n_{\rm P}^{ob} + 1\,
}
for all $\theta\in \Theta_{ob}\cap \Theta_{4}$, where $\R^N\backslash(\Theta_{ob}\cap \Theta_{4})$ is zero-measure. This thus proves ($S2$).

\section{Proof of Theorem \ref{Theorem-GNPI}}

Let $\hat\alpha_j(\theta):\R^N\rightarrow\R$, $j=1,\ldots,n$ be such that
\[
\det\left(s\Ib - \mathbf{O}_{ob}(\theta)^\top\mathbf{O}_{ob}(\theta)\right)
= \sum_{j=1}^{n}\hat\alpha_j(\theta) s^{j-1} + s^{n}\,.
\]
Since the maximal rank of $\mathbf{O}_{ob}(\theta)$ is $n_{ob}^{g}$, it immediately follows that
there exists a $\theta^\ast\in\R^N$ such that $\hat\alpha_{n+1-n_{ob}^{g}}(\theta^\ast)\neq 0$, and $\hat\alpha_j(\theta)=0$ for all $\theta\in\R^N$ and $j\leq n-n_{ob}^{g}$. Note that analytic functions that are not identically zero vanish only on a zero-measure set. This indicates that there is a nonzero-measure set $\Theta_5\subseteq\R^N$ of configuration vector $\theta$ such that the set  $\R^N\backslash\Theta_{5}$ is zero-measure, and the inequality $\hat\alpha_{n+1-n_{ob}^{g}}(\theta)\neq 0$ holds for all $\theta\in\Theta_5$. Thus, we have $\rank(\mathbf{O}_{ob}(\theta))=n_{ob}^{g}$ holds for all $\theta\in\Theta_5$. Recalling that $\R^N\backslash\Theta_{5}$ is zero-measure, this indeed proves that $\rank(\mathbf{O}_{ob}(\theta))=n_{ob}^{g}$ holds for almost all  $\theta\in\R^N$.

With this in mind, we further combine with Proposition \ref{Proposition-NPI} and find that for all configuration $\theta\in\Theta_5$, the resulting network privacy index ${\mathbf I_{rp}} = n - n_{ob}^{g} -1$. Recalling that $\R^N\backslash\Theta_{5}$ is zero-measure, one then can conclude that the network privacy index ${\mathbf I_{rp}} = n - n_{ob}^{g} -1$ holds for almost all configurations $\theta\in\R^N$. This completes the proof.

\end{document}